\documentclass[11pt,letterpaper]{article}
\usepackage{fullpage}

\usepackage{thm-restate}
\usepackage{libertine}
\usepackage[libertine]{newtxmath}

\usepackage{setspace}
\usepackage[breaklinks]{hyperref}

\usepackage[ruled,vlined,linesnumbered]{algorithm2e}
\usepackage{xcolor}
\usepackage{verbatim}
\usepackage{thmtools,thm-restate}


\def\Var{\mbox{Var}}
\def\Max{\mbox{Max}}

\def\prob{\hbox{Pr}}

\newcommand{\E}{\mathbb{E}}

\newcommand{\cI}{{\mathcal I}}

\newcommand{\iden}{{\bf Check-NTSC}\xspace}
\newcommand{\ecov}{{\bf Exact Bounded 3-Cover}\xspace}
\newcommand{\inst}{{\cI}}

\newcounter{note}[section]

\def\prob{\mbox{Pr}}

\newtheorem{theorem}{Theorem}[section]
\newtheorem{lemma}[theorem]{Lemma}

\newtheorem{corollary}[theorem]{Corollary}

\newtheorem{claim}{Claim}[section]

\newtheorem{examp}{Example}[section]

\newenvironment{proof}[1][Proof:]{\begin{trivlist}
\item[\hskip \labelsep {\bfseries #1}]}  {\end{trivlist}}
\newenvironment{definition}[1][Definition:]{\begin{trivlist}
\item[\hskip \labelsep {\bfseries #1}]}{\end{trivlist}}

\newenvironment{remark}[1][Remark:]{\begin{trivlist}
\item[\hskip \labelsep {\bfseries #1}]}{\end{trivlist}}

\newcommand{\qed}{\nobreak \ifvmode \relax \else
      \ifdim\lastskip<1.5em \hskip-\lastskip
      \hskip1.5em plus0em minus0.5em \fi \nobreak
      \vrule height0.4em width0.5em depth0.25em\fi}

\def\prob{\text{Prob}}

\usepackage{amsmath}

\newcommand{\lntsc}{{\frac{\sqrt{n} \log n}{100}}}

\DeclareMathOperator*{\argmax}{arg\,max}

\title{Algorithms for finding $k$ in $k-$means}
\author{
Chiranjib Bhattacharyya\thanks{Dept. of Computer Science and Automation, Indian Institute of Science} \and Ravindran Kannan\thanks{Microsoft Research Lab., India} \and Amit Kumar\thanks{Dept. of Computer Science \& Engg., IIT Delhi} 
}

\newcommand{\rc}[2]{{C^{(#1)}_{{#2}}}}
\newcommand{\outl}{{\textbf {Outlier Centered 1-means}}\xspace}
\newcommand{\calC}{{\cal C}}
\newcommand{\opt}{{\textsf{opt}}}
\newcommand{\hatw}{{\hat w}}
\newcommand{\hatk}{{\hat k}}
\newcommand{\calT}{{\cal T}}

\newcommand{\ntsc}{{\textbf {NTSC}}\xspace}
\newcommand{\wntsc}{{\textbf {weak-NTSC}}\xspace}
\newcommand{\sepn}{{\textbf {well-separatedness}}\xspace}

\newcommand{\centone}{{\textbf {Centered 1-means}}\xspace}

\newcommand{\prune}{{\bf Prune}\xspace}
\newcommand{\hatX}{{\hat X}}
\newcommand{\hmu}{{\hat \mu}}
\newcommand{\hsigma}{{\hat \sigma}}

\begin{document}
\date{}
\maketitle



$k-$means Clustering  requires as input the exact value of $k$, the number of clusters. Two challenges are open: (i) Is there a data-determined definition of $k$ which is provably correct and (ii) Is there a polynomial time algorithm to find $k$ from data ? This paper provides the first affirmative answers to both these questions. As common in the literature, we assume that the data admits an unknown Ground Truth (GT) clustering with cluster centers separated. This assumption alone is not sufficient to answer Yes to (i). We assume a novel, but natural second constraint called no tight sub-cluster (\ntsc) which stipulates that no substantially large subset of a GT cluster can be ``tighter'' (in a sense we define) than the cluster. Our yes answer to (i) and (ii) are under these two deterministic assumptions. We also give polynomial time algorithm to identify $k$. Our algorithm relies on \ntsc to peel off one cluster at a time by identifying points which are tightly packed. 
We are also able to show that 
our algorithm(s) apply to data generated by mixtures of Gaussians and more generally to mixtures of sub-Gaussian pdf's and hence are able to find the number of components of the mixture from data. To our knowledge, previous results for these specialized settings as well, assume generally that $k$ is given besides the data.

\vspace*{3.5in}

\thispagestyle{empty}

\pagebreak

\setcounter{page}{1}
\section{Introduction}

The $k-$means algorithm is widely used in practice in a variety of applications. $k$, the number of clusters is the most basic parameter and we point out later its exact value needs to be known for the algorithm to produce the ``correct'' clustering.

However, there has been a lack of theoretical results on the problem of finding $k$ purely from data.  To our knowledge, the following two fundamental questions remain open for general clustering:
\begin{itemize}
    \item Is there a purely data determined definition of $k$ which is provably correct (a term we define below).
    \item Is there a polynomial time algorithm to find this value, again given no extra information besides the data.
    Data here is the set of points to be clustered.
\end{itemize}
This paper provides the first affirmative answers to both these questions. 
We start with a standard set-up: There is an unknown {\it Ground Truth (GT)} Clustering: a partition of data into subsets $C_1,C_2,\ldots ,C_k$, with the cluster centers separated from each other. The ``correct'' $k$ is the number of clusters in the GT.

We use a notion of cluster-center separation in GT defined below. First, some notation:
For any subset $S$ of data, the mean/center $\mu(S)$ and standard deviation $\sigma(S)$ are defined as usual by:
$$\mu(S)=\frac{1}{|S|}\sum_{x\in S}x\; ;\; \sigma(S)^2=\Max_{v:|v|=1}\frac{1}{|S|} \sum_{x\in S} \left( v\cdot (x-\mu(S))\right)^2.$$
[$\sigma(S)^2$ is the maximum over all directions of the mean squared deviation from the center of $S$.] 
We let $w_0$ be the minimum weight of a GT cluster.

The clusters obey {\it weak separation} if for each $\ell\not=\ell', \ell,\ell'\in [k]$, 
$$|\mu(C_\ell)-\mu(C_{\ell'})|\geq p(1/w_0)\sigma(C_\ell),$$ where, $p$ is a polynomial. This conforms to the adage: Means separated by XX standard deviations.
We will later also use a stronger condition (called strong separation or well-separatedness) which replaces the $\sigma(C_\ell)$ on right hand side above by $\Max_{r\in [k]}\sigma(C_r)$. 

For data generated from a stochastic model, there has been a long study of the minimal separation conditions under which GT can be found. In particular for spherical GMM's, recent deep results have obtained optimal separations~\cite{RV17, KwonC, Hopkins018}. In general, $k$ is assumed to be given even in these special cases. Our aim here is not to restrict to GMM's or in fact to any stochastic model, but GT is to be a deterministic object satisfying certain conditions.

Is (weak) separation a sufficient condition to impose on GT to have a data determined correct value of $k$? I.e.,  is $k$ the minimum number of weakly separated clusters the data can be partitioned into?
[The Occam's razor principle of minimum here excludes the trivial solution $k=n$ (where $n$ is the number of data points)
which clearly satisfies separation, since $\sigma$ is then 0 for each single point cluster.] The answer is no, since the other trivial solution $k=1$ vacuously satisfies separation too. So, we seek additional condition(s) on GT.
We impose some requirements on the condition(s) to strike a good trade-off between how strong they are and how functional they are: 
(i) The conditions must be deterministic, but, 
(ii) the conditions must be satisfied when specialized to data generated according to a GMM with (weak) separation between component means, 
(iii) If the conditions are satisfied by GT, then, the minimum number of clusters, which also satisfy the conditions, into which the data can be partitioned must equal the number of clusters in the GT, 
(iv) The number of clusters in the GT can be found (exactly) in polynomial time.

We formulate a novel condition 
called ``No-Tight-Sub-Cluster'' (\ntsc) 
which together with  (weak) separation satisfies the above requirements. When specialized to stochastically generated data, \ntsc boils down to a natural anti-concentration property of the pdf's of the components of the mixture (satisfied by general Gaussians and all log-concave pdf's).

To motivate \ntsc, consider the special case of determining whether $k=1$. We formulate a clean version of this question and show that even this special case is NP-hard by a reduction from \ecov~\cite{Kann94}.
If $k>1$, then, intuitively, there is a  
subset of data which is more ``tightly packed'' than the whole set. The first try for quantifying ``tightly packed'' would be 1-means cost. A simple 2-component GMM illustrates that this does not work (see Example~\ref{ex:2} in Section~\ref{sec:ex}). 
\ntsc is a new measure of tightness of a subset $T$ of data, which considers $\sigma(T)$ rather than the 1-means cost.
\wntsc : (Informal Definition) We say that a subset $C$ of data satisfies  \wntsc if for every reasonably large subset $S$ of $C$, we have that $\sigma(S)\in\Omega(|S|\sigma(C)/|C|)$. In the special case when $C$ is generated by iid draws from a pdf $f$, we will show that  \wntsc of $C$ follows from the anti-concentration condition: there is a  $1-$dimensional marginal $g$ of $f$, such that for every real $\zeta$, $g(\zeta)\in O(1/\sigma(f))$. 

Later, we will also use a stronger notion denoted \ntsc where we require the \wntsc condition to be satisfied when data is projected onto any 1-dimensional subspace of ${\bf R}^d$.

While we motivated \wntsc by just the $k=1$ or $k>1$ dichotomy, we are now ready to state our first theorem informally, which proves that in fact, \wntsc together with weak separation of cluster means identifies $k$ from the data alone. It does not give a polynomial time algorithm, which we will develop below. We assume that the number of points $n$ is at least  $\frac{100}{w_0^5}. $

\begin{theorem}\label{k-first}
Suppose there is a ground-truth clustering with $k$ clusters which satisfies weak separation  and \wntsc. Then, the minimum $s$ such that there is an $s-$ clustering satisfying \wntsc  equals $k$.
\end{theorem}










\subsection{Our Contributions}
We summarize our main contributions of the paper:
\begin{itemize}
    \item First provable result on determining $k$ from data: Weak Separation and \wntsc are sufficient to determine $k$ (in exponential time) -- see Theorem~\ref{k-first}. 
    \item First polynomial time algorithm to compute $k$: Strong Separation and \ntsc suffice to give us a polynomial time algorithm to compute $k$ from data (Theorem~\ref{main1}). 
    \item Corollary: Since GMM's automatically satisfy \ntsc, if in addition, separation holds, we get a polynomial time algorithm to determine $k$. To the best of our knowledge, there is no earlier explicitly stated provable algorithm for finding $k$ purely from data generated by such GMM's.
    \item First polynomial time algorithm for determining $k$ from data generated by  sub-Gaussian mixtures assuming both separation and anti-concentration (while our anti-concentration condition automatically holds for all log-concave pdf's, it does not hold for arbitrary sub-Gaussian pdf's).
    \item Polynomial time algorithm when $w_0$ is known: Often knowledge of $w_0$ is milder condition than knowing $k$. Under this assumption, we give a polynomial time algorithm for finding $k$ which requires strong separation and \wntsc. In some cases of stochastically generated data (e.g., stochastic block models), we can only show \wntsc, and hence this result is of interest in such settings. 
    \item Besides determining $k$, our results also give a new algorithm to find the approximate clustering. The cluster centers found here are close to the true means, and so can be used as good initialization for $k$-means (see e.g.,~\cite{KK10}). Further, 
     once $k$ is found, GT can also be found approximately by the known algorithm of \cite{AwasthiS12}.
\end{itemize}

\subsection{Informal statement of results and Our Techniques}
In this section, we discuss our results informally starting with an idea of the proof of Theorem (\ref{k-first}): consider a clustering of the points into $s < k$ clusters, say $X_1, \ldots, X_s$. Then one of these clusters, say $X_i$ will contain sufficiently large number of points from two different clusters in the GT -- denote these subsets of $X_i$ as $S_1$ and $S_2$ respectively. Since the GT satisfies weak separation and \wntsc, $\mu(S_1)$ and $\mu(S_2)$ will be sufficiently far apart implying that $\sigma(X_i)$ would be much larger than $\min(\sigma(S_1), \sigma(S_2)$. This will show that the partition under consideration does not satisfy \wntsc producing a contradiction to the hypothesis. 

\noindent
{\bf{From Theorem~\ref{k-first} to a polynomial time algorithm}}

Here, we intuitively describe the challenges in what is left to be done after the theorem.
It suggests an outline of an algorithm: (1)  Starting with $k=1$, try values of $k$ increasing it by 1 each time, (2) For each $k$, find a $k-$clustering, (3)  Check if the clustering satisfies NTSC and if so accept that $k$ and stop.

Both steps 2 and 3 present challenges. For step 2, there are known algorithms~\cite{AwasthiS12, KK10} which will find near optimal $k-$means cost and with means close to the true means. However, there is no proof that for the correct $k$, the clustering so found will satisfy \ntsc (Imagine $O^*(\sqrt{n})$ points from one cluster misplaced into another -- the \ntsc condition considers subsets of this small size). 

For step 3, of course, it is not obvious how to check \ntsc, an intrinsically exponential criterion. In fact, we prove that in general this problem is NP-hard (See Theorem~\ref{thm:nphard}.)


Next we discuss our algorithm. We first assume that we know the minimum relative weight $w_0$ of a cluster. In a low dimensional space, the variance $\sigma(X)$ of a set of points $X$ and the average 1-means cost are close to each other. It turns out that the separation conditions and \ntsc hold if we project to the $1/w_0$-SVD dimensional subspace. Thus our algorithm proceeds as follows: project data to this SVD-subspace, and then peel off points which have low 1-means cost (there are some more subtleties as we don't want the points peeled off to have large $\sigma()$ value).

The other difficulty is that we do not know $w_0$. The  algorithm maintains a guess $\hatw$ for $w_0$ -- it starts with $\hatw$ as 1, and slowly decreases it. For a certain value of $\hatw$, it runs the above-mentioned algorithm. Now one idea would be to check if the resulting clusters, say $X_1, \ldots, X_{\hatk}$, output by the above algorithm satisfy \ntsc (in the SVD-subspace). Although this can be done efficiently, this property may not be satisfied by the clusters produced by the above algorithm when given the correct value $w_0$. Instead we use a more subtle idea: for each of the clusters $X_i$, we prune it by removing subsets which are more tightly clustered than $X_i$. If we end up pruning $X_i$ to less than half its original size, we reject this partitioning (and try a smaller value of $\hatw$). The main technical result here shows that for every $\hatw < w_0$, the resulting clustering $X_1, \ldots, X_{\hatk}$ will always be rejected. The reason is that if $k < \hatk$, then lot of points from two different clusters $X_i$ and $X_j$ belong to a common subset $C_\ell$. But then the means of $X_i$ and $X_j$ cannot be too far (and so we will reject this clustering). If $k$ happens to be larger than $\hatk$, then lot of points from two different clusters $C_i$ and $C_j$ belong to the same set $X_\ell$; but then the pruning procedure above would remove lot of points from $X_\ell$. Thus, we get the following result (note that weak separation has been replaced by a stronger notion, which we call \sepn -- See Section~\ref{sec:prel} for details):
\begin{theorem}
\label{main1}
Let $P$ be a set of points implicitly partitioned into $k$ clusters $C_1, \ldots, C_k$ satisfying well-separetedness and \ntsc. Then there is a polynomial time algorithm to identify the parameter $k$. 
\end{theorem}

Our algorithm also gives an approximate clustering on $P$ into clusters which match with the true clustering on a large fraction of points. 

As an application of this result, we consider points sampled from a mixture of distributions, where each distribution is sub-Gaussian. We assume that the separation between the means of any two distinct distributions from such a mixture is at least $poly(1/w_0)$ times the maximum directional variance of any of the component distributions. Under this mild assumption, we show that the data sampled from the mixture model satisfies well-separatedness and \ntsc. Well-separatedness follows from the fact that the sample and actual means and variances are close to each other. For \ntsc, we need a crucial technical assumption that anti-concentration properties hold for sub-Gaussian pdf's in the mixture. [While all log-concave pdf's (and as a subclass, all Gaussians) automatically satisfy anti-concentration, sub-Gaussian pdf's could behave wildly in sets of small measure and hence, we need the assumption.]

Intuitively, anti-concentration implies that no region of the probability space can have high density, and so the actual samples from this region cannot be more tightly concentrated (compared to rest of the samples from a component distribution).

The anti-concentration property, which relies on upper bounds on the pdf's, does not hold in case of discrete distributions. Weaker versions of anti-concentration which accommodate point masses imply that sampled data satisfy a milder version of \ntsc, which we call \wntsc. Recall that for a point set $X$ to satisfy \ntsc, we needed $\sigma(S)$ to be $\Omega \left( \frac{|S|}{|X|} \sigma(X) \right)$ for every large subset $S$ and restrictions on every line $L$. In \wntsc, we need this property to hold in the underlying space only. 

Our next result gives a polynomial time algorithm for points satisfying \wntsc, but under the assumption that $w_0$ is given in addition to data. 
\begin{theorem}
\label{main2}
Given a set of points $P$ implicitly partitioned into clusters $C_1, \ldots, C_k$ satisfying well-separatedness and \wntsc, and the parameter $w_0$, there is a polynomial time algorithm which correctly identifies $k$.
\end{theorem}
Note that unlike the algorithm in Theorem~\ref{main1}, the above algorithm requires the knowledge of $w_0$, but relies on \wntsc. As mentioned in the introduction, knowledge of $w_0$ is often a much milder assumption than that of $k$. In the algorithm in Theorem~\ref{main1}, we proceeded by peeling off clusters in the SVD subspace. Here we cannot do that because \ntsc may not hold in a subspace. Instead we use a convex program to identify the clusters which are peeled off in each iteration. 

We apply this result to points sampled from stochastic block model (SBM). In SBM, there are $k$ classes with each class $\ell$ having a relative weight $w_\ell$. There is also an implicit $k \times k$ symmetric probability matrix $P$. Points are sampled as follows: first each point is assigned to a class with probability $w_\ell$. Then we build a graph on these points where an edge between two points belonging to communities $i$ and $j$ is added with probability $P_{ij}$. We can view the adjacency matrix as representing points in $\Re^n$. We show that if the intra-cluster probabilities (i.e., diagonal entries of $P$) are sufficiently higher than the inter-cluster probabilities (our separation condition is close to that in many related in works on SBM's (see e.g.~\cite{McSherry01}), the sampled points satisfy well-separatedness and \wntsc. The proof of \wntsc property follows from a weaker anti-concentration result for binomial distributions. 



\subsection{Related Work}
Determining the number of clusters, $k$, is an important open problem which has received considerable attention over the last four decades. 
The first approach for finding $k$ can be traced back to \cite{gauss} where 
clustering is modelled as fitting a mixture of Gaussians with $k$ mixture components. Instead of fitting distributions, Hartigan  \cite{hartigan75} attempted to find the smallest $k$ such that quality of the clustering obtained from a procedure such as $k$-means, is within acceptable limits. Since then several procedures for determing $k$,(see \cite{surveyk} for a survey),
which have shown good empirical performance on clustering data obtained from  large class of distributions have been reported. A common strategy among all such methods is to evolve a measure  of quality of clustering  which have a monotonic behaviour as a function of $k$. Often these measures shows that the monotonic behaviour flattens beyond a certain value of $k$, the ``elbow'', and this value is often taken to be the true value of $k$. A satisfactory explanation of why such methods performs well in practice is still elusive. \emph{Gap Statistic}\cite{gap} is the first rigorous study which aims to explain when such ``elbow'' methods succeed and can be considered as state of the art. \cite{gap} reports that the Gap-statistic performs well 
when the ``clusters are well separated''. However, there are no precise characterizations of the separation, neither the class of distributions for which the Gap Statistic recovers true $k$ is known. 

Lloyd's $k-$means algorithm~\cite{lloyd} is one of the most widely used methods for clustering. 
The algorithm needs $k$ as input. In practice, heuristics are used to get the value of $k$. In the theory/algorithms literature on $k-$means, generally, $k$ is assumed to be given. $k-$means++ algorithm~\cite{ArthurV07} initializes with $k$ means and a wrong value of $k$ can lead it astray at the start.  
Many spectral clustering methods project data points to the $k$ dimensional Singular Value Decomposition subspace ~\cite{VempalaW04, KK10, AwasthiS12} at the outset and an inexact $k$ can make a substantial difference. Recent progress on clustering data generated by GMM's involves iterative algorithms~\cite{RV17, KwonC, Hopkins018, KothariSS18}
starting with an initialization process which again crucially needs $k$ as the number of initial centers to choose. Further all these algorithms require the exact value of $k$.

There has been lot of work on clustering data under deterministic assumptions; however most of these results require the knowledge of the parameter $k$. This is the case for stability defined in ~\cite{BiluL12},
as well as the one introduced in ~\cite{BalcanBG13}.  
If $k$ is not given, even in simple examples, data can satisfy these notions of stability with multiple values of $k$, so $k$ is in general not identifiable under the promise of stability  (see Example~\ref{stability} in Section~\ref{sec:ex}.)
Stability also has another issue: the separation it requires is too large to fit the mold of ``means separated by XX standard deviations'' for the XX we use here. 
~\cite{OstrovskyRSS12} defined a slightly different notion of robustness: the optimal $k$-means objective value of the input is at most $\varepsilon$ times the optimal $k-1$-means objective value. This could be taken as a definition of the right parameter ``$k$'' and is similar to the ``elbow'' method.
However, we show in Section~\ref{sec:ex} that the Elbow method is not sharp enough to figure out the correct value of $k$ when data is generated from (well-separated) mixture of Gaussians. 

A weaker notion of robustness, called the proximity condition is defined in \cite{KK10}.
While this provides motivation for our set-up here with purely deterministic assumptions on GT as well as the use of $\sigma$ (which is related to spectral norm as we see below), \cite{KK10} also require knowledge of $k$ at the outset. The same applies to the improvement due to ~\cite{AwasthiS12}.
 

There has been lot of deep work in clustering data and learning parameters of underlying generative model when the data is generated  from a stochastic model, in particular mixture of distributions. Perhaps the most significant special case is data generated from mixture of Gaussians (GMM). 
 A breakthrough result by Dasgupta~\cite{sanj99} showed that one can recover 
the parameters if the means of component Gaussians are separated. Following this result, a long line of work 
\cite{arora05,Schul,moi,VempalaW04,RV17,KwonC} focusing  on Gaussian Mixture models(GMMs) have developed 
powerful theories which aim to  recover the true parameters even when the separation between the component means is small. All of these assume $k$ is given.

For stochastic block models (SBM's), there have been mathematically sophisticated methods to achieve learning under various separation conditions which rely on gaps between the probability vectors~(see e.g. \cite{McSherry01,Lee2019,Abbe17}).  However, all of these require that $k$ is known.

\subsection{Preliminaries}
\label{sec:prel}
We are given a set of $n$ points $P$ in $\Re^d$. 
These points have an implicit partitioning $C_1, \ldots, C_k$ into $k$ clusters. This partition is the {\it Ground Truth Clustering.}

For a subset $X$ of points in $\Re^d$, we define the maximum directional standard deviation, denoted $\sigma(X)$ as follows: let $n$ denote $|X|$ and $A$ be the $n \times d$ where row $i$ of $A$ is given by $x_i - \mu(X).$ Here $x_i$ is in ${\bf R}^d$ and denotes the $i^{th}$ point in $X$ and $\mu(X)$ denotes the (coordinate-wise) mean of $X$. Then 
$ \sigma(X) := \frac{||A||}{\sqrt{n}}, $ where $||A||$ denotes the spectral norm of $A$. 

We now state the three conditions which we assume are satisfied by Ground Truth clustering: 

\noindent
{\bf{Minimum Cluster Weight Condition:}}
There is a parameter $w_0$ such that $|C_i| \geq w_0 n $ for $i=1, \ldots, k$. We also assume that   $n \geq \frac{100}{w_0^5}. $

\noindent
{\bf{Separation Condition:}}
Let $\sigma_0$ denote $\max_{i=1}^k \sigma(C_i).$
The \sepn (or ``strong separation'') condition states that for every distinct pair of indices $\ell, \ell' \in \{1, \ldots, k\}$, 
\begin{align}
    \label{eq:sept}
|\mu(C_\ell)-\mu(C_{\ell'})|\geq \gamma\sigma_0, \quad \gamma =\frac{K}{w_0^{11}}, \, {\mbox{where  $K$ is  a large enough constant}}.
\end{align}

There is a milder ``weak separation'' condition where the above condition is replaced by 
\begin{align}
    \label{eq:sept1}
|\mu(C_\ell)-\mu(C_{\ell'})|\geq \gamma (\sigma(C_\ell) + \sigma(C_{\ell'})), \quad \gamma =\frac{K}{w_0^{11}}
\end{align}

\noindent
{\bf No Tight Sub-cluster Condition (\ntsc):}
 For $\ell=1,2,\ldots ,k$, every subset $T$ of $C_\ell$ with $|T|\geq \frac{\sqrt{n} \log n}{100}$,
 and any one dimensional subspace $L$, 
\begin{align}
    \label{eq:ntscdef}
\sigma^2(\pi_L(T))\geq \frac{|T|^2}{125 \cdot |C_\ell|^2}\sigma^2(\pi_L(C_\ell)),
\end{align}
where for a set $A$ of points, $\pi_L(A)$, denotes the set of projections of points of $A$ on $L$.

Some of our results will rely on a weaker tightness condition which does not require it to hold on all projections on lines, but only in the original space $\Re^d$.

\noindent
{\bf Weak No Tight Sub-cluster Condition (\wntsc)}
For $\ell=1,2,\ldots ,k$ and for every subset $T$ of $C_\ell$ with $\frac{|T|}{|C_\ell|}\geq
\frac{\sqrt{n} \log n}{100},$ we have 

$$\sigma^2(T)\geq \frac{|T|^2}{125|C_\ell|^2}\sigma^2(C_\ell).$$
It is not hard to show that \ntsc implies \wntsc. 


We now state  a few useful results about $\sigma(X)$ for a set of points $X$. Proof of the following is deferred to the appendix. 
\begin{restatable}{claim}{simple}
\label{cl:sigma1}
Let $S$ be a subset of $X$. Then $|S| \sigma(S)^2 \leq |X| \sigma(X)^2.$
\end{restatable}





The following lemma, whose proof is deferred to the appendix, 
states that if two sets $R$ and $S$ have large intersection, then their means are not too far apart in distance units measured in directional standard deviations $\sigma$. 

\begin{restatable}{lemma}{Tech}
\label{RcapS}
Suppose $R,S\subseteq [n]$. 
Then, 
$$|\mu(R)-\mu(S)|^2\leq \frac{2}{|R\cap S|} \left( |R|\sigma^2(R)+|S|\sigma^2(S)\right).$$
\end{restatable}



\begin{definition}
Given a set of $n$ points $X$ in $\Re^d$, the \centone problem seeks to find a center among $X$ which minimizes the 1-means cost of assigning all of $X$ to this center, i.e., we want to minimize (also denoted as the {\em centered 1-means cost of $X$})
$ \min_{x \in X} \sum_{x' \in X} ||x-x'||^2. $

An instance of the \outl problem is defined as above along with a parameter $m$. The goal is to find a subset $X' \subset X$ of size $m$ such that the centered 1-means cost of $X'$ is minimized. 
\end{definition}

The \centone and the \outl problems can be easily solved in polynomial time because we just need to try each point in $X$ as a potential center (in case of \outl, we just need to pick the $m$ closest points to this center). For an instance $\cI$ of \centone, let $\opt(\cI)$ denote its optimal cost. It is also well-known that $\opt(\cI)$ is at most four times the optimal 1-means cost of the instance $\cI$. We now relate $\opt(\cI)$ to $\sigma(X)$, where $X$ denotes the set of points in $\cI$. The proof is deferred to the appendix. 


\begin{restatable}{claim}{abc}
\label{cl:centone}
  Consider an instance $\cI$ of \centone consisting of a set $X$  of points in $\Re^d$. Then 
  $$ \sigma(X)^2 \leq \frac{\opt(\cI)}{|X|} \leq 4d \cdot \sigma(X)^2. $$
\end{restatable}


We now give an outline of rest of the paper.  In Section~\ref{sec:id1}, we prove Theorem~\ref{k-first}. We give a polynomial time algorithm for identifying $k$ (Theorem~\ref{main1}) in Section~\ref{sec:poly}. This algorithm is presented in two parts -- in Section~\ref{sec:polyw}, we give a polynomial time algorithm which knows the parameter $w_0$. In Section~\ref{sec:polyw1}, we remove this assumption: the algorithm here tries different values of $w_0$ and uses the previous algorithm as a sub-routine. In Section~\ref{sec:polyw2}, we give another polynomial time algorithm which relies on the knowledge of $w_0$, but works with \wntsc (Theorem~\ref{main2}). In Section~\ref{sec:a3}, we apply our results to data generated from stochastic models. In Section~\ref{sec:np}, we show that the problem of finding a subset $X$ of certain size with minimum $\sigma(X)$ is APX-hard. Finally, we give some counterexamples in Section~\ref{sec:ex}. 

\section{A Simple Procedure for Identifying $k$}
\label{sec:id1}
In this section, 
we prove Theorem~\ref{k-first}. We are given a set of $n$ points $P$ in $\Re^d$ satisfying weak separation  and \wntsc.
We also assume that $n \geq \frac{100}{w_0^5}. $ In particular, this implies that 
\begin{align}
    \label{eq:lbnd} 
    w_0^2 n \geq \sqrt{n} \log n
\end{align}


In fact the procedure to identify the parameter $k$ is very simple and is given in Figure~\ref{fig:weak}. We try all partitions of the point set in ascending order of the number of clusters in it, and output the first one which satisfies \wntsc. 


\begin{figure}[h]
    \begin{procedure}[H]
    $\hatk \leftarrow 1.$ \\
    \Repeat{the procedure halts}{

           \For{every partition $\calT = \{T_1, \ldots, T_{\hatk} \}$ of $P$}{
              \If{$\calT$ satisfies \wntsc}{
                    Halt and Ouput $\hatk$.               
              }
           }
          $\hatk \leftarrow \hatk+1$ \\
    
    }
    \end{procedure}
    \caption{ Procedure for identifying $k$ without knowing $w_0$.}
    \label{fig:weak}
\end{figure}

It is easy to see that the procedure will halt with $\hatk \leq k$: when we try the partition $C_1, \ldots, C_k$, it will halt with output $k$. 
%
%
 In order to prove correctness, we need to argue that if our procedure stops earlier with a partition $\calT$, then the number of sets in this partition must be $k$. 

For rest of the argument, fix such a partition $\calT = \{T_1, \ldots, T_{\hatk} \}$  of $P$  which satisfies  \wntsc. We will show that $\hatk = k$. We begin with a simple application of Lemma~\ref{RcapS}:

\begin{lemma}
\label{lem:rcaps}
Consider a cluster $T_\ell$ and a subset $S \subseteq T_\ell, |S| \geq \lntsc$. Then, 
$$ |\mu(S) - \mu(T_\ell)| \leq \frac{50 |T_\ell|^{3/2}}{|S|^{3/2}} \sigma(S). $$
\end{lemma}

\begin{proof}
A direct application of Lemma~\ref{RcapS} shows that 
$$ |\mu(S) - \mu(T_\ell)|^2 \leq \frac{2}{|S|} \left( |T_\ell| \sigma^2(T_\ell) + |S| \sigma^2(S) \right). $$
The \wntsc property 
implies that $\sigma^2(T_\ell) \leq \frac{125 |T_\ell|^2}{|S|^2} \sigma^2(S). $ Substituting this in the above inequality yields the desired result.  \qed
\end{proof}

We first show that no cluster in $\calT$ can have significant number of points from two distinct clusters in $\{C_1, \ldots, C_k\}$. 

\begin{lemma}
\label{lem:clus1}
For any cluster $T_\ell \in \calT$, there is at most one cluster $C_h \in \{C_1, \ldots, C_k\}$ with 
$|C_h \cap T_\ell| \geq \frac{w_0^2 n}{10}$. 
\end{lemma}
\begin{proof}
Consider a cluster $T_\ell \in \calT$, and suppose for the sake of contradiction, there are subsets $S_1, S_2$ of $T_\ell$ and clusters $C_1, C_2$ (by renumbering) such that $S_i \subseteq C_i \cap T_\ell$, $|S_i| \geq \frac{w_0^2 n}{10}$, for $i=1,2$. 

Inequality~(\ref{eq:lbnd}) shows that $|S_i| \geq \lntsc$ for $i=1,2$. Therefore, Lemma~\ref{lem:rcaps} implies that for $i=1,2$: 
$$|\mu(S_i) - \mu(T_\ell)| \leq \frac{2000}{w_0^3} \sigma(S_i). $$

This implies that $$ |\mu(S_1) - \mu(S_2)| \leq \frac{2000}{w_0^3} \left( \sigma(S_1) + \sigma(S_2) \right). $$

Applying Lemma~\ref{lem:rcaps} again, we see that for $i=1,2$:
$$|\mu(C_i) - \mu(S_i)| \leq \frac{50 |C_1|^{3/2}}{|S_1|^{3/2}} \sigma(S_i) \leq \frac{2000}{w_0^3} \sigma(S_i), $$
where the last inequality follows because 
$\frac{|C_i|}{|S_i|} \leq \frac{10 n}{w_0^2  n} \leq \frac{10}{w_0^2}. $ 

Combining the previous two inequalities, we see that 
\begin{align}
\label{eq:ci}
|\mu(C_1) - \mu(C_2)| \leq \frac{6000}{w_0^3} \left( \sigma(S_1) + \sigma(S_2) \right). 
\end{align}

Since $\sigma(S_i) \cdot \sqrt{|S_i|} \leq \sigma(C_i) \cdot \sqrt{|C_i|}$ (Claim~\ref{cl:sigma1}), and we showed above that $\frac{|C_i|}{|S_i|} \leq \frac{10}{w_0^2}, $
it follows that 
$\sigma(S_i) \leq \frac{4}{w_0} \sigma(C_i). $
Substituting this in~\eqref{eq:ci}, we see that 
$$ |\mu(C_1) - \mu(C_2)| \leq \frac{24000}{w_0^4} \left( \sigma(C_1) + \sigma(C_2) \right). $$
But this contradicts the fact that $\{C_1, \ldots, C_k\}$ satisfies well-separatedness. \qed 

\end{proof}

A simple application of the above result is that $\hatk \geq k$, details are deferred to the appendix.
%
 Since $\hatk \leq k$, this shows that $\hatk = k$, and proves correctness of our algorithm.
\section{Polynomial Time Algorithm for Identifying $k$}
\label{sec:poly}

In this section, we prove Theorem~\ref{main1}. As before, we assume that $n$ is $\Omega(1/w_0^5)$. The algorithm is given in two parts. In Section~\ref{sec:polyw}, we describe an algorithm which knows the parameter $w_0$, and then we remove this assumption in Section~\ref{sec:polyw1}.

\subsection{Polynomial Time Algorithm which knows $w_0$}
\label{sec:polyw}

The algorithm, {\bf IdentifyK($P,w_0$)} is outlined in Figure~\ref{fig:new}. 
Let $M$ denote the $d' := \frac{1}{w_0}$-dimensional SVD subspace of $P$. 
For a set of points $X$, let $\mu_M(S)$ denote 
$\mu(\pi_M(S))$, i.e., the mean of the projection of $S$ on $M$. Similarly, let $\sigma_M(S)$ denote $\sigma(\pi_M(S))$.

The algorithm works with the projections $\pi_M(P)$ of the points $P$ on $M$. 
It  runs in iterations -- in iteration $j$ it removes points from a newly discovered cluster $X_j$. Let $P^{(j)}$ denote the set of points at the beginning of iteration $j$. In iteration $j$, we first find a subset $S$ of $w_0 n$ points whose centered  1-means cost (in the projected space) is small -- recall that for point set $X$ and integer $h$, \outl$(X,h$) seeks to find a subset $X' \subseteq X$ of size $h$ with the smallest centered 1-means cost.  The set $X_j$ will be the set of points whose distance from $\mu_M(S)$ in the projected space is within the parameter $r_j$. The algorithm terminates when very few points (at most $w_0 n/10$) remain.

\begin{figure}[h]
    \begin{procedure}[H]
    Initialize $P^{(1)}$ to be the initial set of points $P$. \\
    Define $d' \leftarrow \frac{1}{w_0}. $ \\
    Let $M$ be the $d'$-dimensional SVD subspace of $P$. \\
    \For{$j=1,2,3, \ldots$}{
        Project $P^{(j)}$ to $M$. \\
        Let $S \leftarrow \outl(\pi_M(P^{(j)}), w_0 n/2).$ \\
        Let $X_j$ be the set of points $x \in P^{(j)}$ for which $$|\pi_M(x) - \mu_M(S)| \leq r_j,$$ 
        where \label{l:rj} $$r_j := \frac{2000 k^2 \cdot \sigma_M(S)}{w_0^3} .$$ \\
        Update $P^{(j+1)} \leftarrow P^{(j)} \setminus X_j. $ \\ 
        \If{$|P^{(j+1)}| \leq w_0n/10$}{
          Stop and output $j$ as the estimate for the number of clusters $k$ in the input. \\
        
    }
    }
    \end{procedure}
    \caption{{\bf IdentifyK($P,w_0$)}: Another polynomial time algorithm for identifying $k$}
    \label{fig:new}
\end{figure}


\paragraph{Analysis}
We now analyse the algorithm. 
The following result shows that $\mu_M(C_h)$ and $\mu(C_h)$ are close to each other. 

\begin{lemma}
\label{lem:close2}
For every $h \in \{1, \ldots, k\}$, $|\mu_M(C_h) - \mu(C_h)| \leq \frac{3 \sigma_0}{w_0}$.
\end{lemma}
\begin{proof}

Let $A$ be the matrix whose $i^{th}$ row is given by the coordinates of the $i^{th}$ point in $P$.
Let $C$ be the corresponding matrix 
whose $i^{th}$ row is given by $\mu(C_j)$ where $C_j$ is the cluster containing the corresponding point in $P$. Since $\pi_M(A) - C$ is a matrix of rank at most $d'+k \leq 2d'$
(since $k\leq 1/w_0$)
$$||\pi_M(A) - C||_F^2 \leq \frac{2}{w_0} ||\pi_M(A) - C||^2 . $$

By triangle inequality and properties of SVD (which imply that $\pi_M(A)$ is the best approximation to $A$ of rank at most $d'$) 
$$ ||\pi_M(A) - C|| \leq ||\pi_M(A) - A|| + ||A-C|| \leq 2 ||A-C|| \leq 2\sigma_0 \sqrt{n},$$
where, recall, $\sigma_0$ denotes $\max_{h=1}^k \sigma(C_h).$

Combining this with the above inequality, we see that for any fixed index $h \in \{1, \ldots, k\}$, 
$$ \sum_{x \in \pi_M(C_h)} |x - \mu(C_h)|^2 \leq
||\pi_M(A) - C||_F^2 \leq \frac{8 \sigma_0^2 n}{w_0}. $$

But the mean of $\pi_M(C_h)$ is $\mu_M(C_h)$. Therefore, 
$$ \sum_{x \in \pi_M(C_h)} |x - \mu(C_h)|^2 \geq |C_h| \cdot |\mu(C_h) - \mu_M(C_h)|^2. $$
Since $|C_h| \geq w_0 n$, we get the desired result. \qed
\end{proof}

 At the beginning of iteration $j$, let $\rc{j}{h}$ denote the points of cluster $C_h$ which remain (i.e., $\rc{j}{h} = P^{(j)} \cap C_h$). After suitable relabeling, we assume wlog that the following invariant holds: 

\begin{align}
\label{eq:invnew}
{\mbox{for $h=1, \ldots, j-1$}}, \quad  |\rc{j}{h}| \leq \frac{w_0^2  n}{10}, \\
\label{eq:invnew1}
{\mbox{and for $h=j, \ldots, k$}}, \quad |C_h \setminus \rc{j}{h}| \leq \frac{w_0^2 jn}{10}.  
\end{align}

For $j=1$, the invariant holds trivially because $\rc{1}{h} = C_h$ for all $h$. Now assume that the invariant holds at the beginning of iteration $j$.  
We will prove that it holds at the beginning of iteration $j+1$.
We first upper bound the optimal value of \outl($\pi_M(P^{(j)}),w_0n/2).$
  

\begin{claim}
\label{cl:outl1}
For any $h \in \{j, \ldots, k\}$, the optimum value of $\outl(\pi_M(P^{(j)}), w_0 n/2)$ is at most $4 \sigma_M(C_h)^2 n/w_0$. 
\end{claim} 
\begin{proof}
 Fix an index $h \in   \{j, \ldots, k\}$. One solution to $\outl(\pi_M(P^{(j)}, w_0 n/2)$ is to pick a subset $T$ of $w_0 n/2$ points from $\pi_M(\rc{j}{h})$ -- invariant~\eqref{eq:invnew1} implies that $|\rc{j}{h}| \geq w_0 n/2$. Now, 
 Claim~\ref{cl:centone} implies that the optimal centered 1-means cost of $T$ is at most 
 $$ \frac{4|T|}{d'} \cdot \sigma_M(T)^2 \stackrel{Claim~\ref{cl:sigma1}}{\leq} \frac{4 |C_h|}{w_0} \sigma_M(C_h)^2 \leq \frac{4n}{w_0} \sigma_M(C_h)^2. $$  \qed

\end{proof}

We now show that $\mu_M(S)$ 
(where $S$ is as in Line 6 of the Algorithm of Figure (\ref{fig:new}))
is close to $\mu_M(C_h)$ for some index $h \in \{j, \ldots, k\}$. 
\begin{lemma}
\label{lem:close1}
There is an index $h \in \{j, \ldots, k\}$ such that $|\mu_M(S) - \mu_M(C_h)| \leq 10 \sqrt{k} \cdot \sigma_M(C_h)/w_0. $ Further, $|S \cap C_h| \geq \frac{w_0 n}{4k}. $
\end{lemma}

\begin{proof}
Invariant~\eqref{eq:invnew1} implies that $|S \cap \cup_{h=j}^k C_h| \geq w_0 n/4$. Therefore, there is an index $h \in \{j, \ldots, k\}$ such that 
$|S \cap C_h| \geq \frac{w_0 n}{4k}$.  Suppose $|\mu_M(S) - \mu_M(C_h)| > \frac{10 \sqrt{k} \cdot \sigma_M(C_h)}{w_0}$. Using $|a-b|^2 \geq \frac{a^2}{2} - b^2$ for any real $a,b$, we get:
\begin{align*}
\mbox{1-means cost of }\pi_M(S)&\geq\\
 \sum_{x \in \pi_M(S_h)} |x-\mu_M(S)|^2 & \geq \sum_{x \in \pi_M(S_h)} \left( \frac{|\mu_M(C_h) - \mu_M(S)|^2}{2} - |x-\mu_M(C_h)|^2  \right) \\
 & >  |S_h| \cdot \frac{50k \cdot \sigma_M(C_h)^2}{w_0^2} - \frac{\sigma_M(C_h)^2 n}{w_0}.
\end{align*}
Since $|S_h| \geq \frac{w_0 n}{4k}$, it follows that the 1-means cost (and hence the centered 1-means cost) of $\pi_M(S)$ is more than $\frac{4 \sigma_M(C_h)^2 n}{w_0}$, which contradicts Claim~\ref{cl:outl1}. \qed
\end{proof}

After renumbering, we can assume that the index $h$ in Lemma~\ref{lem:close1} is $j$. Hence
\begin{equation}\label{789}
|S\cap C_j|\geq w_0n/(4k).
\end{equation}
We now argue that for every index $h\geq j+1$, $\mu_M(C_h)$ and $\mu(C_h)$ are close to each other.

\begin{claim}
\label{cl:close2}
For every $h \in \{j+1, \ldots, k\}$, $|\mu_M(S) - \mu_M(C_h)| \geq \gamma \sigma_0/4$. 
\end{claim}
\begin{proof}
The well-separatedness  condition and Lemma~\ref{lem:close2} imply that  for any $h \in \{j_+1, \ldots, k\}$, 
$$ |\mu_M(C_h) - \mu_M(C_j)| \geq \frac{\gamma \sigma_0}{2}. $$
The result now follows from Lemma~\ref{lem:close1}. \qed
\end{proof}

We now  relate $\sigma_M(S)$ to $\sigma_M(C_j)$. 

\begin{lemma}
\label{lem:rel}
$$ \frac{w_0 \cdot \sigma_M(C_j)}{100 \cdot k^{3/2}} \leq \sigma_M(S) \leq \frac{3 \cdot \sigma_M(C_j)}{w_0}. $$
\end{lemma}

\begin{proof}
Claim~\ref{cl:outl1} implies that the centered 1-means cost of $\pi_M(S)$ is at most $\frac{2 \sigma_M(C_j)^2 n}{w_0}$. Since $|S| = w_0 n/2,$ Claim~\ref{cl:centone} implies that $\sigma_M(S)^2 \leq \frac{4 \sigma_M(C_j)^2}{w_0^2}. $
For the other direction, we use \ntsc. Since $S_j:= S \cap C_j$ is a subset of $C_j$, and $|S_j| \geq \frac{w_0 n}{4k} \geq \lntsc$ (by inequality~(\ref{eq:lbnd})), it follows from \ntsc that 
$$ \sigma_M(S_j)^2 \geq \frac{w_0^2}{2000 k^2} \sigma_M(C_j)^2. $$
Further $S_j \subseteq S$, and so, Claim~\ref{cl:sigma1} implies that $\sqrt{|S_j|} \sigma_M(S_j) \leq \sqrt{|S|} \sigma_M(S). $ Since $|S_j| \geq \frac{w_0 n}{4k}$ and $|S| = w_0 n/2$,  $|S_j| \geq \frac{|S|}{2k}$.  So, we get 
$\sigma_M(S_j) \leq \sqrt{2k} \cdot \sigma_M(S)$. Using this in the above inequality yields
$$ \sigma_M(S)^2 \geq \frac{w_0^2}{4000 k^3} \sigma_M(C_j)^2. $$
This proves the desired result.  
\qed
\end{proof}

\begin{claim}
\label{cl:elem}
Let $X$ be a set of $m$ points in $M$. For every $\alpha > 0$, at least $\left(1 - \frac{1}{\alpha^2 \cdot w_0} \right) m$ points of $X$ lie within distance $\alpha \cdot \sigma(X)$ of $\mu(X)$. 
\end{claim}
\begin{proof}
Since $M$ is $1/w_0$-dimensional subspace, the 1-means cost of $X$ is at most $\frac{\sigma(X)^2 \cdot m}{w_0}$. The result now follows from a simple averaging argument. \qed
\end{proof}

The algorithm defines the following parameter (line~\ref{l:rj}) $$r_j := \frac{2000 k^2 \cdot \sigma_M(S)}{w_0^3} .$$ Let $X_j$ be the set of points $x \in P^{(j)}$ for which $|\pi_M(x) - \mu_M(S)| \leq r_j$ (as defined in line~\ref{l:rj}). Note that the algorithm knows $\sigma_M(S)$ and so it can compute $r_j$.  The following key result shows that $X_j$ is very close to $C_j$. 

\begin{lemma}
\label{lem:key1}
At most $\frac{w_0^2 n}{10}$ points of $C_j$ lie in $P^{(j)} \setminus X_j$. Further, at most $\frac{w_0^2 n}{10k}$ points of $C_h$, $h \in \{j+1, \ldots, k\}$ lie in $X_j$.  
\end{lemma}
\begin{proof}
Lemma~\ref{lem:rel} shows that $r_j \geq \frac{20 \sqrt{k} \cdot \sigma_M(C_j)}{ w_0^2}$. 
Lemma~\ref{lem:close1} now shows that all points $x \notin X_j$ satisfy: 
$$ |\pi_M(x) - \mu_M(C_j)| > \frac{10 \sqrt{k} \cdot \sigma_M(C_j)}{ w_0^2}. $$
Claim~\ref{cl:elem}  implies that number of points in $C_j$ which do not belong to $X_j$ is at most $\frac{w_0^3 |C_j|}{100 k} \leq \frac{w_0^2 n}{10k}$. This proves the first part of the result. 

For the second part, fix an index $h \in \{j+1, \ldots, k\}$. 
Lemma~\ref{lem:rel} shows that $r_j \leq \frac{6000 k^2 \sigma_0}{w_0^4} \leq \gamma \sigma_0/8$. Claim~\ref{cl:close2} now shows that if $x \in C_h \cap X_j$, then $|\pi_M(x) - \mu_M(C_h)| \geq \frac{\gamma \sigma_0}{8} \geq \frac{\gamma \sigma_M(C_h)}{8}$. Claim~\ref{cl:elem} now shows that $|X_j \cap C_h| \leq \frac{64n}{\gamma^2 w_0} \leq \frac{w_0^3 n}{10k}$. This proves the desired result. \qed
\end{proof}
The above result shows that the invariant conditions are satisfied if we define $P^{(j+1)} := P^{(j)} \setminus X_j$. It also follows that 
$$|P^{(k+1)}| \leq \sum_{h=1}^k |\rc{k+1}{h}| \leq \frac{w_0^2 k n}{10} \leq \frac{w_0 n}{10}. $$
Therefore, the algorithm will stop at the end of iteration of $k$. 
This proves the correctness of our algorithm.

Finally, we prove some results which will be useful in the next section. 
\begin{lemma}
\label{lem:useful}
$\sigma_M(X_h) \leq  \frac{4000 k^2}{w_0^4} \cdot \sigma_M(C_j). $
\end{lemma}
\begin{proof}
Since all points in $X_h$ are within distance $r_j$ of $\mu_M(S)$, it follows that 
$\sigma_M(X_h) \leq r_j$. The result now follows from Lemma~\ref{lem:rel} and the definition of $r_j$. \qed
\end{proof}

\begin{lemma}
\label{lem:useful1}
Let $h,h' \in \{1, \ldots, k\}$ be two distinct indices. Then 
$$ |\mu_M(X_h) - \mu_M(X_{h'})| \geq \frac{800}{w_0^4} (\sigma_M(X_h) + \sigma_M(X_{h'})). $$
\end{lemma}

\begin{proof}
First consider $X_h$. Let $S_h$ be the set $S$ considered in iteration $h$ of our algorithm. Since all points in $\pi_M(X_h)$ lie within radius $r_h$ of $\mu_M(S_h)$, we see that 
$$ |\mu_M(X_h) - \mu_M(S_h)| \leq r_h. $$
Combined with Lemma~\ref{lem:close1} and Lemma~\ref{lem:rel}, this implies that 
$$|\mu_M(X_h) - \mu_M(C_h)| \leq 2r_h. $$
Similarly, $|\mu_M(X_{h'}) - \mu_M(C_{h'})| \leq 2 r_{h'}. $ Therefore, 
$$ |\mu_M(X_h) - \mu_M(X_{h'})| \geq |\mu_M(C_h) - \mu_M(C_{h'})| - 2r_h -2r_{h'} 
\geq \frac{\gamma \sigma_0}{2} - 2r_h -2r_{h'}, $$
where the last inequality follows from Lemma~\ref{lem:close2} and well-separatedness. 
By Lemma~\ref{lem:rel}, $r_h, r_{h'} \leq \frac{2000 k^2\sigma_0}{ w_0^4}$. The above inequality now implies that 
$$ |\mu_M(X_h) - \mu_M(X_{h'})| \geq \frac{\gamma \sigma_0}{4}. $$
The desired result now follows because $\sigma_M(X_h) \leq r_h \leq \frac{2000 k^2\sigma_0}{ w_0^4} \leq \frac{w_0^4 \gamma \sigma_0}{3200}$, and similarly for $\sigma_M(X_{h'})$. 
\qed
\end{proof}
\subsection{Polynomial Time Algorithm without the knowledge of $w_0$}
\label{sec:polyw1}

The algorithm in the previous section assumed that we know $w_0$. In this section, we show how the algorithm can be modified to  work even when $w_0$ is unknown.
The idea is to maintain an estimate $\hatw$ for $w_0$, which starts with 1 and decreases in steps of size $\frac{1}{n}$. For a given $\hatw$, we can run the algorithm {\bf IdentifyK} described in the previous section -- this would lead to a disjoint partition $X_1, \ldots, X_{\hatk}$ of a large enough subset of $P$. One idea would be to check that all of these subsets $X_j$ satisfy \ntsc. Since we don't know how to check \ntsc efficiently, we could try to check this for the projection $\pi_M(P)$ (here, the spectral norm and the 1-means cost are close to each other upto a factor depending on $w_0$ only). The problem with this approach is that this test would fail even if we guessed the right value of $w_0$. Indeed, the sets $X_1, \ldots, X_k$ constructed by {\bf IdentifyK} could have a  non-negligible fraction (e.g. $O(w_0^2)$) of points from clusters other their respective representative ones. To rectify this issue, we first define a {\bf Prune} procedure which, given a set of points $X$ and parameter $\hatw$,  shaves off subsets  of $X$ that are tighter than $X$ (in the subspace $M$). Given such a procedure, we now run it on each of the sets $X_1, \ldots, X_{\hatk}$ returned by {\bf IdentifyK($P, \hatw)$}, and check that none of these sets shrink by a large factor. 




\paragraph{The \prune procedure}
Given a subspace $M$, a set of points $X$ and parameter $\hatw$, the \prune procedure reduces $X$ to a subset $\hatX$ as shown in Figure~\ref{fig:prune}. 

\begin{figure}[h]
    \begin{procedure}[H]
    Initialize $\hatX \leftarrow X$. \\
    \Repeat{$\hatX$ does not change}{
  		Call a subset $T$ of $\hatX$ to be {\em tight} if it satisfies the following two conditions: \\
  		\quad (a) $|T| \geq \lntsc$. \\
  		\quad (b) Optimal average centered 1-means cost of $\pi_M(T)$ is less than $\frac{\hatw^{12} \cdot |T|^2 \cdot \sigma_M(X)^2}{c \cdot |X|^2}$, where $c = 10^{12}. $ \\
  		 \If{$\hatX$ has a tight subset $T$}{
  		  $\hatX \leftarrow \hatX \setminus T. $\\
  		  }
  		   
    }
    \end{procedure}
    \caption{ \prune($X, M, \hatw$)}
    \label{fig:prune}
\end{figure}
%

\subsubsection{The Algorithm}
The algorithm is described in Figure~\ref{fig:unknownw}. It maintains an estimate $\hatw$ for $w_0$. Initially, the estimate starts at 1, and decreases in steps of $1/n$. We can assume wlog that $w_0$ is an integral multiple of $1/n$ (since we can always scale it up to the nearest such multiple). In each such iteration (with a guess $\hatw$), it calls {\bf IdentifyK($P,\hatw$)}. If the clusters returned by this procedure satisfy the given conditions, it halts with and outputs $\hatk$. 

\begin{figure}[h]
    \begin{procedure}[H]
    \label{l:for}
    Initialize $\hatw \leftarrow 1. $ \\
    \Repeat{ the algorithm stops with an estimate $\hatk$ }{ 
           \label{l:call}
           Call {\bf IdentifyK($P,\hatw$)}. \\
           Let $X_1, \ldots, X_{\hatk}$ be the clusters found by it. \\
           Let $\hatX_h$ be the set returned by \prune($X_h, M, \hatw$), $h=1, \ldots, \hatk$. \\
           Check the following conditions: \\
           (a) for all distinct pairs $h,j$, $1 \leq h,j \leq \hatk$, $$ |\mu_M(X_h) - \mu_M(X_j)| \geq \frac{800}{\hatw^4} \left( \sigma_M(X_h) + \sigma_M(X_j) \right). $$ \\
           (b)  for each $h, 1 \leq h \leq \hatk$, $|\hatX_h| \geq |X_h|/2$ \\
           (c) for each $h, 1 \leq h \leq \hatk$, $|X_h| \geq \hatw n/2$. 
           
           \If{ all the above conditions are satisfied}{
            halt and output $\hatk$ as  the number of clusters. \\
            }
            \Else{
            Decrease $\hatw$ by $1/n$. \\
            }

    }
    \end{procedure}
    \caption{ Polynomial time algorithm for identifying $k$ without knowing $w_0$.}
    \label{fig:unknownw}
\end{figure}


\paragraph{Case $\hatw = w_0$}

We begin by first showing that if $\hatw = w_0$, then the clustering $X_1, \ldots, X_k$ produced by {\bf IdentifyK($P,\hatw$)} will satisfy conditions~(a),~(b) and~(c). Lemma~\ref{lem:useful1} shows that condition~(a) will be satisfied. Condition~(c) is satisfied by Lemma~\ref{lem:key1}. 
We now proceed to show that condition~(b) will be satisfied as well. 

\begin{lemma}
\label{lem:tight}
Let $X_1, \ldots, X_k$ be the subsets produced by {\bf IdentifyK($P,\hatw$)}. Suppose $T$ is a tight subset of $X_h$ for an index $h \in \{1, \ldots, k\}$. Then at most half of the points in $T$ belong to $C_h$. 
\end{lemma}
\begin{proof}
Suppose not. Let $T'$ be $T \cap C_h$, and so $|T'| \geq |T|/2$. It follows from Claim~\ref{cl:sigma1} that $\sigma_M(T)^2 \geq \sigma_M(T')^2/2$. Since $T' \subseteq C_h$ and $|T'| \geq \lntsc$, \ntsc implies that $$ \sigma_M(T')^2 \geq \frac{|T'|^2}{125 \cdot |C_h|^2} \sigma_M(C_h)^2 \geq 
\frac{|T|^2}{500 \cdot |X_h|^2} \sigma_M(C_h)^2,$$
where the last inequality follows from the fact that $|X_h| \geq |C_h|/2$ (by Lemma~\ref{lem:key1} the two sets differ in at most $\frac{w_0^2 n}{10}$ elements). 

Now, Lemma~\ref{lem:useful} implies that 
$$ \sigma_M(T)^2 \geq \frac{\sigma_M(T')^2}{2} \geq \frac{ w_0^8 \cdot |T|^2}{c k^4 \cdot |X_h|^2} \sigma_M(X_h)^2 
\geq \frac{ w_0^{12} \cdot |T|^2}{c  \cdot |X_h|^2} \sigma_M(X_h)^2 $$
which contradicts the fact that $T$ is a tight subset of $X_h$ (using Claim~\ref{cl:centone}). \qed
\end{proof}

\begin{corollary}
\label{cor:tight}
Let $X_1, \ldots, X_k$ be the subsets produced by {\bf IdentifyK($P,\hatw$)}. For all $h \in \{1, \ldots, k\}$, $|\hatX_h| \geq |X_h|/2$, where $\hatX_h$ is the set returned by \prune$(X_h, M, \hatw)$. 
\end{corollary}
\begin{proof}
We know that at most $\frac{w_0^2 n}{10}$ points of $X_h$ lie outside $C_h$ (Lemma~\ref{lem:key1}). Whenever the \prune procedure removes a subset $T$ from $X_h$, at least $|T|/2$ elements belong to $X_h \setminus C_h$ (by Lemma~\ref{lem:tight}). Therefore, it can remove at most $\frac{w_0^2 n}{10}$ elements from $X_h$, which is at most $|X_h|/2$. \qed
\end{proof}

In order to prove correctness, it remains to show that if the algorithm stops before $\hatw$ reaches $w_0$, then it returns $\hatk = k$. So assume that the algorithm stops at a value $\hatw > w_0$ and let $X_1, \ldots, X_{\hatk}$ be the corresponding clusters, which satisfy conditions~(a), (b), (c). The proof proceeds in two parts: we first show that $\hatk \geq k$, and then show that $\hatk \leq k$. 

\paragraph{Case $\hatk < k$:}

We first consider the case when $\hatk < k$. 
\begin{claim}
\label{cl:part}
There is an index $h \in \{1, \ldots, \hatk\}$ and distinct indices $\ell_1, \ell_2 \in \{1, \ldots, k\}$ such that $|X_h \cap C_{\ell_1}|, 
|X_h \cap C_{\ell_2}| \geq \frac{w_0^2 n}{2}. $
\end{claim}
\begin{proof}
Suppose not. Then for every $h \in \{1, \ldots, \hatk\}$, there is at most one index, call it $\ell_h$, such that $|X_h \cap C_{\ell_h}| \geq \frac{w_0^2 n}{2}$. Since $k > \hatk$, there is an index $\ell$ which is not equal to $\ell_h$ for any $h \in \{1, \ldots, \hatk\}$. But then
$$|C_\ell| = \sum_{h=1}^\hatk |C_\ell \cap X_h| \leq \hatk \cdot \frac{w_0^2 n}{2} < \frac{w_0^2 k n}{2} \leq \frac{w_0 n}{2}, $$
which is a contradiction. \qed
\end{proof}

Let $h$ be the index guaranteed be Claim~\ref{cl:part}, and by renumbering assume without loss of generality that the indices $\ell_1, \ell_2$ are $1,2$ respectively. 
We now show that $\sigma_M(X_h)$ is large. 

\begin{lemma}
\label{lem:larges}
$\sigma_M(X_h)^2$ is at least $\frac{w_0^3 \cdot \gamma^2 \sigma_0^2}{1600}$.
\end{lemma}
\begin{proof}
Let $T_i$ denote $X_h \cap C_i, i = 1,2$.
By our assumption, $\frac{|T_i|}{|C_h|} \geq w_0^2/2,$ $i=1,2.$
Lemma~\ref{RcapS} along with Claim~\ref{cl:sigma1} imply that 
$$ |\mu_M(C_i) - \mu_M(T_i)|^2 \leq \frac{4 |C_i| \cdot \sigma_0^2 }{|T_i|}
\leq  \frac{8 \sigma_0^2}{w_0^2}. $$
Now, Lemma~\ref{lem:close2} along with the separation condition implies that 
$$|\mu_M(T_1) - \mu_M(T_2)| \geq \frac{\gamma \sigma_0}{10}. $$
Therefore one of $|\mu_M(T_1) - \mu_M(X_h)|, |\mu_M(T_1) - \mu_M(X_h)|$, say the former, is at least $\frac{\gamma \sigma_0}{20}$. 

Let $B$ be the $|X_h| \times d$ matrix whose $j^{th}$ row   given by $\pi_M(x_j) - \mu_M(X_h)$, where $x_j$ denotes the coordinates of $j^{th}$ point in $X_h$.
Then 
$$||B||_F^2 \geq \sum_{x \in T_1} |\pi_M(x) - \mu_M(X_h)|^2 \geq |T_1| \cdot |\mu_M(T_1) - \mu_M(X_h)|^2 \geq \frac{\gamma^2 \sigma^2}{400} \cdot |T_1| \geq 
\frac{w_0^2 \cdot \gamma^2 \sigma^2 n}{800}. $$
Since $B$ has rank at most $\frac{2}{w_0}$, it follows that 
$$\sigma_M(X_h)^2 = ||B||^2 \geq \frac{w_0^3 \cdot \gamma^2 \sigma^2}{1600}. $$ \qed
\end{proof}

\begin{lemma}
\label{lem:prune}
For any index $\ell \in \{1, \ldots, k\}$, $\hatX_h$ contains at most $\frac{w_0^2 n}{2}$ elements of $C_\ell$. 
\end{lemma}
\begin{proof}
Suppose, for the sake of contradiction, that $X':=\hatX_h \cap C_\ell$ has size more than $\frac{w_0^2 n}{2} \geq 2d \log d$ for some $\ell \in \{1, \ldots, k\}$. We first upper bound $\sigma_M(X')^2$. 

Since $X'$ is a subset of $C_\ell$, Claim~\ref{cl:sigma1} shows that 
$$\sigma_M(X')^2 \leq \frac{|C_\ell|}{|X'|} \sigma_M(C_\ell)^2 \leq \frac{2\sigma_0^2}{w_0^2}. $$

Claim~\ref{cl:centone} shows that  the average centered 1-means cost of $\pi_M(X')$ is at most $\frac{4\sigma_0^2}{w_0^3}$. 

Now Lemma~\ref{lem:larges} implies that 
$$ \frac{\hatw^{12} |X'|^2 \cdot \sigma_M(X)^2}{c \cdot |X|^2} \geq 
\frac{w_0^{12} \cdot w_0^4 \cdot w_0^3 \cdot \gamma^2 \sigma_0^2}{3200 c} 
\geq  \frac{2\sigma_0^2}{w_0^3}. $$
But then the \prune($X_h, M, \hatw$) procedure should have removed $X'$ from $\hatX_h$, a contradiction. \qed
\end{proof}

\begin{corollary}
\label{cor:prune}
$|\hatX_h| < |X_h|/2$.
\end{corollary}
\begin{proof}
By Lemma~\ref{lem:prune}, 
$$|\hatX_h| \leq \frac{w_0^2 k n}{2} \leq \frac{w_0 n}{2} \leq \frac{\hatw n}{2} \leq |X_h|/2. $$ \qed
\end{proof}

Corollary~\ref{cor:prune} shows that $\{X_1, \ldots, X_{\hatk}\}$ violate condition~(b), which is a contradiction. Therefore, $\hatk \geq k$. 

\paragraph{Case  {${\mathbf \hatk > k}$}}:

We now assume $\hatk > k$. 
\begin{claim}
\label{cl:part1}
There is an index $\ell \in \{1, \ldots, k\}$ and distinct indices $h_1, h_2 \in \{1, \ldots, \hatk\}$ such that $|C_\ell \cap \hatX_{h_1}|, 
|C_\ell \cap \hatX_{h_2}| \geq \frac{\hatw^2 n}{4}. $
\end{claim}
\begin{proof}
Suppose not. So every index $\ell \in \{1, \ldots, k\}$, there is at most one index, say $h_\ell$ for which $|C_\ell \cap \hatX_{h_\ell}| \geq \frac{\hatw^2 n}{4}$. But then there is an index $h \in \{1, \ldots, \hatk\}$ which is not of the form $h_\ell$ for any $\ell \in \{1, \ldots, k\}$. Therefore, 
$$|\hatX_h| = \sum_{\ell=1}^k |\hatX_h \cap C_\ell| \leq \frac{k \hatw^2 n}{4} \leq \frac{\hatk \hatw^2 n}{4} \leq \frac{\hatw n}{4}, $$
which is a contradiction. \qed

\end{proof}

Let $\ell$ be the index guaranteed by Claim~\ref{cl:part1}, and assume by renumbering that $h_1, h_2 = 1,2$ respectively. Let $X_i'$ denote $C_\ell \cap \hatX_i, i=1,2$. 
By applying Lemma~\ref{lem:rcaps} to $X_i', C_\ell$, we see that for $i=1,2$, 
\begin{align*}
|\mu_M(X_i') - \mu_M(C_\ell)| \leq \frac{200}{\hatw^3} \sigma_M(X_i')
\end{align*}
Claim~\ref{cl:sigma1} implies that $\sigma_M(X_i') \leq \sqrt{\frac{|X_i|}{|X_i'|}} \sigma_M(X_i) \leq 
\frac{2 \sigma_M(X_i)}{\hatw}. $ Therefore, we get for $i=1,2$, 
$$|\mu_M(X_i') - \mu_M(C_\ell)| \leq \frac{400}{\hatw^4} \sigma_M(X_i)$$

So we get 
\begin{align}
\label{eq:id}
|\mu_M(X_1') - \mu_M(X_2')| \leq \frac{400 \cdot  (\sigma_M(X_1)+
\sigma_M(X_2))}{\hatw^{4}}
\end{align}

%
By Lemma~\ref{RcapS} and Claim~\ref{cl:sigma1}, for $i=1,2$, 
$$|\mu_M(X_i) - \mu_M(X_i')|^2 \leq \frac{4 |X_i| \sigma_M(X_i)^2}{|X_i'|} 
\leq \frac{16 \sigma_M(X_i)^2}{\hatw^2}. $$
Combining the above with~\eqref{eq:id}, we see that 
$$|\mu_M(X_1)  - \mu_M(X_2)| \leq \frac{800 \cdot  (\sigma_M(X_1)+
\sigma_M(X_2))}{\hatw^4}, $$
which contradicts~(a). 
This shows that $\hatk = k$, and proves the correctness of our algorithm. 

\section{Polynomial Time Algorithm with \wntsc}
\label{sec:polyw2}

In this section, we give another polynomial time algorithm for identifying $k$. This algorithm requires the knowledge of $w_0$, but relies on \wntsc only. The algorithm is shown in Figure~\ref{fig:newalgo}. It runs in several iterations, and in each iteration it finds a large subset $X$ of points which is close (in Hamming distance) to a newly discovered cluster. It removes this set $X$ from further consideration and repeats the whole process till very few points remain. 

In a particular iteration $j$, it first projects the remaining points $P^{(j)}$ to the subspace $M$. In this subspace, the algorithm finds a subset $S$ of size $w_0n/2$ with minimum centered 1-means cost.
In line~\ref{l:convex}, we use a convex program $\calC$ which has three parameters: $\calC(m, \mu,T)$, where $m$ is a positive integer, $\mu$ is a point and $T$ is a subset of $P$. The convex program seeks to find a subset $T', |T'|=m$, of  $T$ which is close to $x$ and has low $\sigma(T')$ value. The details are described below. 
The algorithm calls $\calC(m, \mu(S), P^{(j)})$ for values of $m$ starting from $w_0 n/2$ and increasing in unit steps. It stops at a value $m^\star$ of $m$ when the value of the convex program becomes much higher than that when $m$ was $w_0n/2$. Now it takes the solution $\calC(m^\star, \mu(S), P^{(j)})$ and rounds it an integral solution, which yields the desired subset subset $X$ of $P^{(j)}$.

\begin{figure}[h]
    \begin{procedure}[H]
    Initialize $P^{(1)}$ to be the initial set of points $P$. \\
    Define $d' \leftarrow \frac{1}{w_0}. $ \\
    Let $M$ be the $d'$-dimensional SVD-subspace of $P$. \\
    \For{$j=1,2,3, \ldots$}{
        Let $S \leftarrow \outl(\pi_M(P^{(j)}), w_0 n/2).$ \label{l:outl} \\
        Define $\nu_j \leftarrow \mu(S)$. \\
         Consider the convex program $\calC(m, \nu_j, P^{(j)})$.  \label{l:convex} \\
        Let $m^\star \geq w_0 n/2$ be the highest index $m$ such that $$\opt(\calC(m, \nu_j, P^{(j)})) \leq \frac{72000}{w_0^{3.5}} \cdot  \opt(\calC(w_0n/2, \nu_j, P^{(j)}))$$ \label{l:mst} \\
        Let $y$ be the (fractional solution) to $\calC(m^\star, \nu_j, P^{(j)}). $ \\
        Use Lemma~\ref{lem:rounding} to round $y$ to an integral solution $y'$. \label{l:round} \\
        Let $X \subseteq P^{(j)}$ be the set of points $i$ for which $y'_i = 1$. 
\label{l:rem}        
        \\
        Update $P^{(j+1)} \leftarrow P^{(j)} \setminus X. $\\
        \If{$|P^{(j+1)}| \leq w_0n/10$}{
          Stop and output $j$ as  the number of clusters  in the input. \\
        
    }
    }
    \end{procedure}
    \caption{{\bf IdentifyKnew($P,w_0$)}: Polynomial time algorithm for identifying $k$}
    \label{fig:newalgo}
\end{figure}

We now describe the convex program $\calC(m, \mu, T)$, where $|T| \geq m \geq w_0 n$. For each point $x_i \in T$, we have a variable $y_i \in [0,1]$. Define a $|T| \times d$ matrix $B_y$ as follows: the $i^{th}$ row of $B_y$ is $y_i(x_i - \mu)$ (and hence is a linear function of $y_i$). The convex program is: 

\begin{align}
\min. \quad & \frac{||B_y||}{\sqrt{m}} \notag \\
\sum_{x_i \in P^{(j)}} y_i & = m \\
1 \, \geq y_i & \geq 0 \quad \quad \forall x_i \in T. 
\end{align}

Note that this is a valid convex program since $||B_y||$ is a convex function of the entries in $B_y$. In line~\ref{l:round}, we refer to a rounding algorithm for a solution $y$ to this convex program. We describe this in the result below. 
\begin{lemma}
\label{lem:rounding}
Consider a fractional solution $y$ to $\calC(m, \mu, T)$, where  $w_0 n \leq m \leq |T|$. Then there is an integral solution $y'$ to the convex program such that (i) $||B_{y'}|| \leq \frac{20 \cdot ||B_y||}{w_0^2},$ and (ii) $\sum_{x_i \in T} y'_i \geq m - \frac{w_0^2 n}{20}. $
\end{lemma}
\begin{proof}
Let $a$ be the number of points $x_i$ for which $y_i \geq \frac{w_0^2}{20}, $ and $b$ the number of remaining points in $T$. Then 
$$m = \sum_{x_i \in T} y_i \leq a + \frac{w_0^2 b}{20} \leq a + \frac{w_0^2 n}{20}, $$
where the last inequality follows from the fact that $|T| \leq n$. 
Therefore, $a \geq m - \frac{w_0^2 n}{20}$. Now, we define $y_i' = 1$ if $y_i \geq \frac{w_0^2}{20},$ and 0 otherwise. Clearly $||B_{y'}|| \leq \frac{20 \cdot ||B_y||}{w_0^2}, $ because omitting rows from a matrix can only decrease its spectral norm, and scaling its entries by a factor $\alpha$ scales the spectral norm by $\alpha$ as well. This proves the lemma.  \qed
\end{proof}

This completes the description of the algorithm. We now analyse it. 

\paragraph{Analysis}
As in the analysis in Section~\ref{sec:polyw1}, we  write down the invariant conditions that will be satisfied at the beginning of each iteration. 
At the beginning of iteration $j$, let $\rc{j}{h}$ denote the points of cluster $C_h$ which remain in $P^{(j)}$. After suitable relabeling, we assume that the following invariant holds: 

\begin{align}
\label{eq:inv}
{\mbox{for $h=1, \ldots, j-1$}}, \quad  |\rc{j}{h}| \leq \frac{w_0^2  n}{10}, \\
\label{eq:inv1}
{\mbox{and for $h=j, \ldots, k$}}, \quad |C_h \setminus \rc{j}{h}| \leq \frac{w_0^2 jn}{10}.  
\end{align}

 For $j=1$, the invariant holds trivially because $\rc{1}{h} = C_h$ for all $h$. Now assume that the invariant holds at the beginning of iteration $j$. 
 %
Let $A^{(j)}$ be the matrix whose  $i^{th}$ row is given by the coordinates of the $i^{th}$ point in $P^{(j)}$.
Similarly, let $C^{(j)}$ be the corresponding sub-matrix of $C$ obtained by retaining only those rows corresponding to the points in $P^{(j)}$.
 Clearly $||A^{(j)} - C^{(j)}|| \leq ||A - C||. $ 
 Let $\pi_M(x)$ denote the projection of a point $x$ on $M$. Similarly, for a matrix $B$ of suitable dimension, let $\pi_M(B)$ denote the matrix obtained by projecting each row of $B$ on $M$. The following claim is known~\cite{KK10}, the proof is given for sake of completeness.  

\begin{claim}
\label{cl:svd}
$$ ||\pi_M(A^{(j)}) - C^{(j)} ||_F^2 \leq \frac{8 \sigma_0^2n}{w_0} $$
\end{claim}
\begin{proof}
Since $(\pi_M(A^{(j)}) - C^{(j)})$ has rank at most $\frac{2}{w_0}$, 
$$ ||\pi_M(A^{(j)}) - C^{(j)} ||_F^2 \leq \frac{2}{w_0} ||\pi_M(A^{(j)}) - C^{(j)} ||^2. $$
Triangle inequality and the definition of $M$ now imply
$$||\pi_M(A^{(j)}) - C^{(j)} || \leq ||\pi_M(A^{(j)}) - A^{(j)} || + ||A^{(j)} - C^{(j)} || \leq 2 ||A^{(j)} - C^{(j)} || \leq 2 ||A-C|| = 2 \sigma_0 \sqrt{n}, $$
where the second inequality follows from the fact that $\pi_M(A^{(j)}$ is the best rank -$\frac{1}{w_0}$ approximation to $A^{(j)}$, and the ranl of $C^{(j)}$ is at most $k \leq \frac{1}{w_0}$. 
This proves the desired result. 
\qed
\end{proof}

The following observation follows easily from Claim~\ref{cl:svd}. 
\begin{corollary}
\label{cor:svd}
The optimum value of the instance $\outl(\pi_M(P^{(j)}), w_0 n/2)$ is at most 
$\frac{16 \sigma_0^2 n}{w_0}$.  
\end{corollary}
\begin{proof}
The invariant~(\ref{eq:inv1}) implies that $\rc{j}{h} \geq w_0n/2$ for any $h \in  \{j, \ldots, k\}$. Fix such an index $h$.  Consider a solution to $\outl(\pi_M(P^{(j)}), w_0 n/2)$  consisting of a subset $\pi_M(X)$ of $w_0 n/2$ points from $\rc{j}{h}$. The centered 1-means cost of this solution is at most 
$$ 2 \sum_{x \in X} |\pi_M(x) - \mu(C_h)|^2 \leq 2||\pi_M(A^{(j)}) - C^{(j)} ||_F^2, $$
and now the desired result follows from Claim~\ref{cl:svd}. \qed
\end{proof}


Let $S$ be the solution to $\outl(\pi_M(P^{(j)}, w_0n/2)$ as in line~\ref{l:outl} in the algorithm. Corollary~\ref{cor:svd} implies that the centered 1-means cost of $S$ is at most  $\frac{16 \sigma_0^2 n}{w_0}$.  We now show that $\mu(S)$ is a good approximation to the mean of one of the clusters $C_j, \ldots, C_k$. 

\begin{lemma}
\label{lem:centerapprox}
There is an index $h \in \{j, \ldots, k\}$ such that $|\mu(S) - \mu(C_h)| \leq \frac{20 \sigma_0}{w_0}$. 
\end{lemma}
\begin{proof}
For sake of contradiction, let us assume that the statement of the lemma is false. The invariant~(\ref{eq:inv}) implies that $P^{(j)}$ has at most $w_0 n/4$ points from $C_1 \cup \ldots \cup C_{j-1}$. Therefore $S$ contains at least $w_0 n/4$ points from $C_j \cup \ldots \cup C_k$ -- let $S_h$ denote the set of points in $S \cap C_h, h = j, \ldots, k$. Then the 1-means cost (and hence, the centered 1-means cost) of $S$ is at least 
\begin{align*}
 \sum_{h=j}^k \sum_{x \in S_h} |x-\mu(S)|^2 & \geq \sum_{h=j}^k \sum_{x \in S_h} \left( \frac{|\mu(C_h) - \mu(S)|^2}{2} - |x-\mu(C_h)|^2 \right) \\
 & \geq \sum_{h=j}^k \frac{200 \sigma_0^2 |S_h|}{w_0^2} - ||\pi_M(A^{(j)}) - C^{(j)}||_F^2 \\
 & \geq \frac{50 \sigma_0^2 n}{ w_0} - \frac{8 \sigma_0^2 n}{w_0} = \frac{42 \sigma_0^2 n}{w_0},
 \end{align*}
 where the first inequality uses the fact that $(a-b)^2 \geq a^2/2 - b^2$ for any real $a,b$; the second inequality uses the fact that $|\mu(S) - \mu(C_h)| \geq \frac{20 \sigma_0}{w_0}$, and third inequality uses Claim~\ref{cl:svd} and the fact that $\sum_{h=j}^k |S_h| \geq w_0n/4$. But now we get a contradiction because the centered 1-means cost of $S$ is at most  $\frac{16 \sigma_0^2 n}{w_0}$. \qed
\end{proof}

By relabeling, we can assume that the index $h$ in Lemma~\ref{lem:centerapprox} is $j$ (and so, $\mu(C_j)$ is closest to $\mu(S)$ among $\mu(C_j), \ldots, \mu(C_h)$). Thus, we have a good estimate $\nu_j := \mu(S)$ for $\mu(C_j)$. Having found this estimate $\nu_j$, we go back to the original space and find most of the points of $C_j$. 

We now give upper and lower bounds on  $\opt(\calC(m,\nu_j, P^{(j)} ))$. We begin with the simpler upper bound. 
\begin{lemma}
\label{lem:upper}
Let $m$ be a value in the range $[w_0n/2, |\rc{j}{j}|]$. Then $\opt(\calC(m,\nu_j, P^{(j)})) \leq \frac{2\sigma(C_j)}{\sqrt{w_0}} + |\mu(C_j) - \nu_j|. $
\end{lemma}
\begin{proof}
We pick any set $X$ of $m$ points in $\rc{j}{j}$ and define $y_i = 1$ for these points. 
Let $A'$ be the matrix where each row is the coordinates of a unique point in $X$, and $C'$ be the matrix with each row being $\mu(C_j)$. Then 
$$ ||B_y|| \leq ||A'-C'|| + \sqrt{m} \cdot |\mu(C_j) - \nu_j| \leq \sigma(C_j) \sqrt{n}  + \sqrt{m} \cdot |\mu(C_j) - \nu_j|. $$
This proves the desired result. \qed
\end{proof}

We now prove lower bounds on $\calC(m, \nu_j, P^{(j)})$. Before that we bound $||B_y||$ for certain kind of integral solutions $y$. 
\begin{lemma}
\label{lem:lowery}
 Let $y$ be an integral solution to $\calC(m, \nu_j, P^{(j)}), m \leq n,$ with 
$y_i=1$ for at least $\frac{w_0^2 n}{20k}$ points $x_i \in C_h$ for an index $h \neq j$. Then the objective value of this solution is at least $\frac{\gamma \sigma_0 w_0}{40 \sqrt{k}}. $
\end{lemma}
\begin{proof}
Consider such a solution $y$. Let $X$ be the set of points $x_i$ for which $y_i = 1$, and let $h \neq j$ be the index  such that $X_h := X \cap C_h$ has at least $\frac{w_0^2 n}{20k}$ points. Let $C'$ and $D'$ be $|X_h| times d$ matrices with each row being $\mu(C_h)$ and $\nu_j$ respectively. Let $A_h'$ be the $|X_h| \times d$  matrix with  row $i$  containing the coordinates $x_i$ of the  $i^{th}$ point in $X_h$. 
Now
$$||B_y|| \geq ||A_h'-D|| \geq ||C'-D'|| - ||A_h' - C'|| \geq \sqrt{|X_h|} \cdot |\mu(C_h) - \nu_j| - ||A-C||. $$

The well-separatedness condition along with Lemma~\ref{lem:centerapprox} implies that 
$|\mu(C_h) - \nu_j| \geq \frac{\gamma \sigma_0}{2}$. Therefore, we get 
$$ ||B_y|| \geq \frac{w_0 \gamma \sigma_0 \sqrt{n}}{20 \sqrt{k}} - \sigma_0 \sqrt{n} 
\geq    \frac{w_0 \gamma \sigma_0 \sqrt{n}}{40 \sqrt{k}}. $$
Since $m \leq n$, the desired result follows. \qed
\end{proof}

\begin{lemma}
\label{lem:lower}
The optimal value of $\calC(w_0n/2, \nu_j, P^{(j)})$ is at least 
$$ \frac{w_0^3 |\mu(C_j) - \nu_j|}{72000} + 
\frac{ w_0^3 \sigma(C_j)}{24000}$$
\end{lemma}
\begin{proof}
Let $m$ denote an integer between $w_0 n/2$ and $w_0 n/2 - \frac{w_0^2 n}{10} \geq w_0 n/4$. By Lemma~\ref{lem:rounding}, it suffices to show that any integral solution to $\calC(m,\nu_j, P^{(j)}\nu_j, P^{(j)})$ has objective function value at least 
\begin{align}
\label{eq:int}
\frac{w_0 |\mu(C_j) - \nu_j|}{3600} + 
\frac{ w_0 \sigma(C_j)}{1200}
\end{align}
 We proceed to show this next.

Consider an arbitrary integral solution $y$ to $\calC(m,\nu_j, P^{(j)})$. Let $X$ denote the set of points $i$ for which $y_i = 1$. Since $|X| \geq w_0 n/4$, and invariant~(\eqref{eq:inv}) shows that $P^{(j)}$ has at most $w_0 n/8$ points from $C_1 \cup \ldots \cup C_{j-1}$, it follows that $\sum_{h=j}^k |X_h| \geq \frac{w_0 n}{8}$, 
where $X_h$ denotes $X \cap C_h$. Consider an index $h \in \{j+1, \ldots, k\}$. We claim that $|X_h| \leq \frac{w_0 n}{16k}$. Indeed, otherwise  
 Lemma~\ref{lem:lowery} shows that the objective function value of this solution is at least (using Lemma~\ref{lem:centerapprox})
$$ \frac{\gamma \sigma_0 w_0}{40 \sqrt{k}} \geq \frac{w_0 |\mu(C_j) - \nu_j|}{3600} + 
\frac{ w_0 \sigma(C_j)}{1200}$$ 
and so we are done. 

Therefore, for the rest of the argument, we can assume that $|X_j| \geq \frac{w_0 n}{16}. $ Let $B'$ be the submatrix of $B_y$ consisting of only those rows corresponding to points in $X_j$. Clearly, $||B_y|| \geq ||B'|| \geq \sigma(X_j) \cdot \sqrt{|X_j|}.$ Therefore
$$ \frac{||B_y||}{\sqrt{m}} \geq \frac{\sigma(X_j) \sqrt{|X_j|}}{\sqrt{m}} \geq \frac{\sigma(X_j)}{3}. $$
Since $|X_j| \geq \frac{\sqrt{n} \log n}{100}$ (by~\eqref{eq:lbnd}),  \wntsc implies that $\sigma(X_j) \geq \frac{w_0 \sigma(C_j)}{190}. $
Thus, we get 
\begin{align}
\label{eq:l1}
 \frac{||B_y||}{\sqrt{m}} \geq \frac{ w_0  \sigma(C_j)}{570}. 
\end{align}

We now give a second lower bound on $||B_y||$. 
 Let $A_j'$ be the matrix where each row $i$ is $x_i$ for a unique point $x_i \in X_j$.
Let $C'$ and $D'$ be matrices (with as many rows as $|X_j|$) where each row is $\mu(C_j)$ and $\nu_j$ respectively. Then  
$$||B_y|| \geq ||A_j'- D'|| \geq   ||C'-D'|| - ||A_h'-C'|| \geq  |\mu(C_j) - \nu_j| \sqrt{|X_j|}- \sigma(C_j) \sqrt{n}. $$
Therefore, 
\begin{align*}
\frac{||B_y||}{\sqrt{m}} \geq \frac{|\mu(C_j) - \nu_j|}{3}  - \frac{2\sigma(C_j)}{\sqrt{w_0}}. 
\end{align*}

Using the bound on $\sigma(C_j)$ from~\eqref{eq:l1} in the inequality above and rearranging, we get 
\begin{align}
\label{eq:l2}
\frac{||B_y||}{\sqrt{m}} \geq \frac{w_0 |\mu(C_j)-\nu_j|}{1800}
\end{align}

Taking the average of~\eqref{eq:l1} and~\eqref{eq:l2} yields~\eqref{eq:int}. 
\qed
\end{proof}

We are now ready to prove the key result. 

\begin{theorem}
\label{thm:poly}
For any $m \in [w_0n/2, |\rc{j}{j}|]$, $\opt(\calC(m, \nu_j, P^{(j)})) \leq \frac{72000}{w_0^{3.5}} \cdot  \opt(\calC(w_0n/2), \nu_j, P^{(j)}). $ 
\end{theorem}
\begin{proof}
The statement follows from Lemma~\ref{lem:upper} and Lemma~\ref{lem:lower}. \qed
\end{proof}

In our algorithm, 
we find the highest $m^\star \geq w_0n/2$ such that $\opt(\calC(m^\star), \nu_j, P^{(j)})$ is at most $\frac{72000}{w_0^{3.5}} \cdot  \opt(\calC(w_0n/2), \nu_j, P^{(j)})$ (line~\ref{l:mst}). Theorem~\ref{thm:poly} implies that $m^\star \geq |\rc{j}{j}|$. Let $y$ be the optimal solution to $\calC(m^\star, \nu_j, P^{(j)})$. Using Lemma~\ref{lem:rounding}, we round $y$ to an integral solution $y'$. Let $X$ denote the set of points $x_i$ for which $y'_i = 1$. We define $\rc{j+1}{h}$ as $\rc{j}{h} \setminus X$. We now show that this preserves invariants~\eqref{eq:inv} and~\eqref{eq:inv1}. 

\begin{claim}
\label{lem:inv1}
For any index $h \neq j$, $|X \cap C_h| \leq \frac{w_0^2 n}{20k}. $ Therefore, $\rc{j+1}{h}$ satisfies invariant~\eqref{eq:inv1} for $h = j+1, \ldots, k$. 
\end{claim}
\begin{proof}
From Lemma~\ref{lem:centerapprox} and Lemma~\ref{lem:upper}, objective value of $y'$ is at most 
$$\frac{1.6 \times 10^{6}}{w_0^{5.5}} \cdot \left( \frac{2\sigma(C_j)}{\sqrt{w_0}} + |\mu(C_j) - \nu_j| \right) \leq \frac{4 \cdot 10^{7} \cdot \sigma_0}{w_0^{6.5}}, $$
where the last inequality follows from Lemma~\ref{lem:centerapprox}. 
Let $h \neq j$ be an index such that $|X \cap C_h| \geq \frac{w_0^2 n}{10k}$. Then Lemma~\ref{lem:upper} shows that 
the objective value of $y'$ is at least $\frac{\gamma \sigma_0 w_0}{40 \sqrt{k}}$, a contradiction. 

Since $X$ includes at most $\frac{w_0^2 n}{20}$ points of $C_h$, $h > j+1$, invariant~\eqref{eq:inv1} follows. \qed
\end{proof}

For $h \leq j-1$, invariant~\eqref{eq:inv} holds because the same holds for $\rc{j}{h}$. So it remains to show that $|\rc{j+1}{j}| \leq \frac{w_0^2 n}{10k}$. Lemma~\ref{lem:rounding} shows that $|X| \geq m^\star - \frac{w_0^2 n}{20}  \geq |\rc{j}{j}| - \frac{w_0^2 n}{20}$. Lemma~\ref{lem:inv1} shows that $X$ can include at most $\frac{w_0^2 n}{20}$ points from clusters other than $C_h$. So, $|X \cap C_h| \geq |\rc{j}{j}| - \frac{w_0^2 n}{10}$, and so, $|\rc{j+1}{j}| \leq \frac{w_0^2 n}{10}. $
This shows that the invariant conditions hold at the beginning of iteration $j+1$ also. 

Observe that after $k$ iterations, at most $w_0n/10$ points will remain, and so the procedure will stop after exactly $k$ iterations. This proves correctness of our algorithm.

\section{Applications to Stochastic Models}
\label{sec:a3}
In this section, 
we prove that data from several stochastic models of data generation satisfy our deterministic conditions and so our results about finding the parameter $k$ apply here. We  consider data generated from a mixture $\sum_{\ell=1}^kw_\ell  F_\ell$ of pdf's $F_\ell$, where $\sum_\ell w_\ell = 1$. 

First we give some definitions associated with a pdf $F$. 
Given a pdf $F$ in $\Re^d$, let $\hmu(F)$ denote the mean of $F$. Define $\hsigma(F)$ as the maximum directional variance of $F$, i.e., $$ \hsigma(F) := \max_{v: |v|=1} \Var_{X \sim F}[X \cdot v], $$
where $X \sim F$ denotes a random vector $x$ distributed according to $F$. 
For a pdf $F$ and a unit vector $u$, let $F_u$ denote the 1-dimensional marginal density of $F$ along $u$, i.e., the density corresponding to the random variable $u \cdot X,$ where $X \sim F$. 
\begin{definition}\label{sc-F}
The sample complexity of $F$ denoted sc$(F)$ is the minimum integer $s$ such that for all $m\geq s$, if $S$ is a set of $m$ iid samples drawn according to $F$, with high probability, $S$ satisfies:
\begin{align}
|\mu(S)-\hmu(F)|&\leq \gamma \hsigma(F)\label{mu-empirical}\\
\forall \mbox{\ unit vectors  $v\in {\bf R}^d$}, \quad \quad  \sigma_v(S)&\leq 2\hsigma(F_v)\label{sigma-empirical}
\end{align}
\end{definition}
\begin{definition}
We say that $F$ satisfies {\bf anti-concentration}, if for all unit vectors $u \in \Re^d$, 
\begin{align}
    \label{anti-conc}
F_u(\zeta) \leq \frac{4}{\hsigma(F_u)}, \quad \forall \zeta \in \Re. 
\end{align}
\end{definition}
We note that Gaussians and indeed any log-concave distribution satisfies anti-concentration property~\cite{LLS07}.


We now state the main theorem on the application our deterministic results to stochastic mixtures. After proving it, we give one class of examples - mixtures of sub-gaussian distributions.

\begin{theorem}\label{stoc-appl}
Suppose $F=\sum_{\ell=1}^kw_\ell F_\ell$ is a mixture of pdf's on ${\bf R}^d$ with $w_\ell\geq 2w_0$ and $\hsigma_0=\Max_\ell \hsigma(F_\ell)$ satisfying 
\begin{equation}\label{F-ell-sep}
|\hmu(F_\ell)-\hmu(F_{\ell'})|\geq 6 \gamma\hsigma_0\forall \ell\not=\ell'.
\end{equation}
Suppose we pick $n$ iid samples from $F$, where,
$$n\geq 100 \log k \cdot \Max_\ell \mbox{sc}(F_\ell)/w_0.$$
Further suppose $F_1,F_2,\ldots ,F_k$ satisfy anti-concentration property.
Let $C_\ell$ denote the samples picked according to $F_\ell$. 
Then the clustering $\{C_1, \ldots, C_k\}$ satisfies the minimum cluster weight condition w.r.t. $w_0$, well-separatedness condition~(\ref{eq:sept}) and \ntsc~(\ref{eq:ntscdef}). 
\end{theorem}
\begin{proof}
First, by H\"offding-Chernoff, 
it follows that whp, $|C_\ell| \geq w_0 n$, and 
$|C_\ell|\geq \mbox{sc}(F_\ell)$ for all $\ell \in [k]$. 
Conditions~(\ref{mu-empirical}) and (\ref{sigma-empirical}) hold by the 
  definition of sample complexity.
  Conditions~(\ref{mu-empirical}) and~(\ref{F-ell-sep}) imply that for all distinct $\ell, \ell' \in [k]$, 
  $$ |\mu(C_\ell) - \mu(C_{\ell'})| \geq 2 \gamma \hsigma_0. $$
  Condition~(\ref{sigma-empirical}) implies that $\hsigma_0 \geq \hsigma_2$, and so the well-separatedness condition for $C_\ell, \ell \in [k]$ follows.


To prove \ntsc, we use a simple V-C dimension-based argument:
Let $u$ be any unit vector.
Let $T$ be a subset of $C_\ell$ with $|T|\geq \sqrt n\ln n/100$. 
We use anti-concentration to prove a lower bound on $\sigma_v(T)$. 
Let $\zeta=\frac{|T|\sigma_u(F_\ell)}{12|C_\ell|}$.
 By anti-concentration,
$$\prob_{X \in F_{\ell}}(|u\cdot X-u\cdot\mu(T)|\leq\zeta)\leq \frac{4}{\hsigma_u(F_u)} \cdot 2 \zeta = 
\frac{2|T|}{3|C_\ell|}.$$

Consider the set system consisting of intervals on the line, where the measure of each interval is given by $F$. Let $\varepsilon$ denote $\frac{|T|}{10 |C_\ell||}.$
As long as $|T| \geq \sqrt{n} \ln n,$, 
$$ |C_\ell| \geq \frac{c'}{\varepsilon^2} \ln \frac{k}{ \varepsilon}, $$
where $c'$ is a large enough constant. It follows that 
with probability at least $1 - \frac{o(1)}{k}$, $C_\ell$ is an $\varepsilon$-sample for this set system. Let $I$ denote the interval of length $\zeta$ around $u \cdot \mu(T)$. Then we have shown that the measure of $I$ is at most $\frac{2|T|}{3 |C_\ell|}.$ Therefore, 
$$\frac{|C_\ell \cap I|}{|C_\ell|} \leq \frac{2|T|}{3 |C_\ell|}+ \frac{2|T|}{10 |C_\ell|}. $$
Therefore, at  least $0.23|T|$ points of $T$ are at distance greater than $\zeta$ from $u \cdot \mu(T)$, showing that $\sigma_L(T)\geq .033\sigma_L(F_\ell)|T|/|C_\ell|$.
Now, using (\ref{sigma-C-vs-F}), we have $\sigma_L(F_\ell) \geq \sigma_L(C_\ell)/2$ , thus proving \ntsc. \qed
\end{proof}

\subsection{Sub-Gaussian Densities}
\label{sec:subg}
In this section, we consider the class of sub-gaussian pdf's
(general Gaussians are a special class) and bound their sample complexity and hence prove that they satisfy the hypothesis of Theorem (\ref{stoc-appl}).

We refer to \cite{vers} for an introduction to sub-Gaussian random vectors and their sub-Gaussian norm. Briefly, for a real-valued sub-Gaussian random variable $X$, the sub-Gaussian norm of $X$, denoted $||X||_{\mbox{sg}}$ is defined as
$$||X||_{\mbox{sg}}=\mbox{Sup}_{p\geq 1}(E|X|^p)^{1/p}/\sqrt p.$$ 
($X$ is sub-gaussian iff this exists.) For a sub-Gaussian random vector $X$, its sub-Gaussian norm is the supremum over all unit vectors $v$ of the sub-Gaussian norm of $v\cdot X$. If $X$ is a random Gaussian vector with variance-covariance matrix $\Sigma$, then, its sub-Gaussian norm is
$\sqrt {||\Sigma||}$. We will use what we call ``shape parameter'' $\kappa$ defined as: For a sub-Gaussian random vector $X$ with non-singular variance-covariance matrix $\Sigma$, 
the {\it shape parameter} $\kappa(X)$ is defined by
$$\kappa(X)=\frac{||X||_{\mbox{sg}}}{\lambda^{1/2}_{\mbox{min}}(\Sigma)}.$$
If $G$ is a Gaussian pdf, $\kappa(G)$ is the square root of the condition number of the variance-covariance matrix.
If $G$ is the pdf of random variable $X$, we also write $||G||_{\mbox{sg}}$ for $||X||_{\mbox{sg}}$.
\begin{theorem}
Suppose $F=\sum_{\ell=1}^kw_\ell F_\ell$ is a mixture of sub-Gaussian pdf's satisfying anti-concentration 
and the separation condition (\ref{F-ell-sep}), 
with $w_\ell\geq w_0$. Let $\kappa_0=\Max_\ell \kappa(F_\ell)$. Suppose 
$$n\geq c\kappa_0^4d^2/w_0,$$
and a set $S$ of $n$ iid samples are drawn $F$. Letting $C_\ell$ be the subset of $S$ drawn accroding to $F_\ell$, we have whp: $C_\ell$ satisfy the well-separatedness and \ntsc.
\end{theorem}
\begin{proof}
The following Lemma bounds the sample complexity of the $F_\ell$. 
By H\"offding-Chernoff it follows that $|C_\ell|\geq \mbox{sc}(F_\ell)$ whp. Now, Theorem (\ref{stoc-appl}) implies the current Theorem.\qed 
\end{proof}
\begin{lemma}
If $F$ is a sub-gaussian pdf on ${\bf R}^d$, we have
$$\mbox{sc}(F)\leq 100\kappa^4(F)d^2.$$
\end{lemma}
\begin{proof}
Let $m$ be any integer with $m\geq 100\kappa^4(F)d^2$ and let $C$ be a set of $m$ iid sample drawn according to $F$.

By concentration of real-valued sub-Gaussian random variables (see the H\'offding inequality in Proposition 5.10 of \cite{vers}), for each $i\in [d]$, for all $t\geq 0$, 
$$\prob(|\mu_i(F)-\mu_i(C)|>t)\leq \exp \left(1-\frac{t^2n}{4\hsigma^2(F)} \right).$$
Put $t=10\sqrt{\ln dk}\hsigma(F)/\sqrt n$ and use union bound over all $i \in [d], \ell \in [k]$ to get
$$\prob( \forall \ell \in [k]: |\mu(F)-\mu(C)| \leq t\sqrt d) = 1 - o(1), $$
proving that whp (\ref{mu-empirical}) holds.

Next, we prove that (\ref{sigma-empirical}) holds.
Since the second moment is minimum when centered at the mean we have, for a unit vector $v$, 
\begin{equation}\label{F-ell-C-ell}
\sigma_v(C)^2\leq \frac{1}{|C|}\sum_{x\in C}(v\cdot (x-\mu(F)))^2=\frac{1}{|C|}v^TA^{T}Av,
\end{equation}
where, $A$ is a $|C| \times d$ matrix with each row of the form 
$x-\mu(F)$ for an $x\in C$.

We use Theorem 5.39 and Remark 5.40 of \cite{vers}, which state that whp the following holds (with $\Sigma$ being the  variance-covariance matrix of $F$):
$$||\frac{1}{|C|}A^{T}A-\Sigma||
\leq \frac{c||F||_{\mbox{sg}}^{2}
\sqrt d}{\sqrt {m}}.$$ 
From this, noting that $||F||_{\mbox{sg}}\leq\kappa_0\hsigma_v(F)$ for all $v$,
we get using (\ref{F-ell-C-ell}) that with whp, for all unit vectors $v$,
\begin{align*}
&\sigma_v^2(C_\ell)
\leq (1+\kappa^2(F)\sqrt{d/m})\hsigma_v^2(F)
\end{align*}
Now, using the lower bound on $m$, we get whp
\begin{equation}
\label{sigma-C-vs-F}
\forall \, {\mbox {unit vectors $v$}}: \quad  \sigma_v^2(C)\leq 3\hsigma_v^2(F).
\end{equation}
\qed

\end{proof}

\subsection{Stochastic Block Models}
In the stochastic block model, there are $k$ communities and an unknown $k \times k$ symmetric matrix $P$ with each entry in the range $[0,1]$. Each community $\ell$ has a relative weight $w_\ell$ such that $\sum_{\ell \in [k]} w_\ell=1. $ A graph $G$ on $n$ vertices is sampled from  this model as follows: each vertex $v$ first chooses a community with probability proportional to its weight. Conditioned on this event, an edge appears between two vertices belonging to communities $i$ and $j$ respectively with probability $P_{ij}$. Given this graph, we would like to recover the parameter $k$ (and the partitioning of $G$ into communities). We assume that for each $\ell \in [k]$, $P_{\ell \ell} = \max_{\ell' \in [k]} P_{\ell \ell'}$ -- this is a natural assumption in this setting because we want intra-community density of edges to be higher than inter-community density of edges. We also assume that the probability matrix $P$ satisfies the following condition on separation between inter-cluster and intra-cluster probabilities: for every distinct $\ell, \ell' \in [k]$
\begin{align}
    \label{eq:sbmc}
    \frac{(P_{\ell \ell} - P_{\ell \ell'})^2}{P_{\max}} \geq \frac{400 \max(\gamma^2, \log n/w_0^2)}{ n},
\end{align}
where $P_{\max}$ denotes $\max_{\ell \in [k]} P_{\ell \ell}$. 
This is similar to the separation condition used in several works on SBMs (see e.g. ~\cite{McSherry01}). 
 Also, we assume $P_{\ell \ell} \leq 1/2$ for all $\ell$ -- again this is a standard assumption in such settings because the sampled graphs are not dense. 

Let the vertices of $G$ be labelled $v^1, \ldots, v^n$. For each $i \in [n]$, we associate a vector $w^i \in \Re^n$ as follows: the coordinate $w^i_j$ is 1 if $(v^i,v^j)$ is an edge, 0 otherwise. The cluster $C_\ell$ corresponds to all the vectors $w^i$, where $v^i$ belongs to community $\ell$. Let $w_0 = \min_\ell w_\ell/2$. A It follows that whp each of the clusters $C_\ell$ has at least $w_0 n$ points. 
It remains to show that the clusters  $C_\ell$ satisfy well-separatedness and \wntsc.
We first condition on the choice of community for each vertex $v^i$ (subject to each cluster size being at least $w_0 n$) and show that these two conditions are satisfied with high probability (and so these two conditions hold whp even when we remove this conditioning).

We associate a distribution $F_\ell$ in $\Re^n$ for each community $\ell$ as follows: a random variable $X \sim F_\ell$ has the coordinate $X_i$ distributed independently as $B(P_{\ell \ell'})$, where $B(p)$ denotes Bernoulli distribution with parameter $p$, and $\ell'$ denotes the community to which $v^i$ belongs. The following claim is easy to see. 

\begin{claim}
\label{cl:sbm}
For each $\ell \in [k]$, $\hmu(F_\ell)_i = P_{\ell \ell'}$ where $v^i$ belongs to community $\ell'$; and $P_{\ell \ell}/2 \leq \hsigma(F_\ell)^2 \leq P_{\ell \ell}. $ 
\end{claim}
\begin{proof}
The result on $\hmu(F_\ell)$ is easy to see. We now prove the second statement. Fix a unit vector $v$. Let $X \sim F_\ell$. Then 
$$ \Var[\sum_{i \in [n]} X_i v_i] = \sum_{i \in [n]} v_i^2 \Var[X_i] \leq \sum_{i \in [n]} v_i^2  P_{\ell \ell}  = P_{\ell \ell}. $$
\qed
\end{proof}

\begin{claim}
\label{cl:sbm1}
For every distinct $\ell, \ell' \in [k]$, 
$$|\hmu(F_\ell) - \hmu(F_{\ell'})| \geq 20 \max(\gamma, \sqrt{\log n/w_0}) \sqrt{P_{\max}}. $$ 
\end{claim}
\begin{proof}
By Claim~\ref{cl:sbm}, 
$$ |\hmu(F_\ell) - \hmu(F_{\ell'})|^2 \geq w_0 n (P_{\ell \ell} - P_{\ell \ell'})^2 \geq 400 \max(\gamma^2, \log n/w_0^2) P_{\max}, $$
by~\eqref{eq:sbmc}. 
\end{proof}


Since each entry of a vector $w^i$ is an iid Bernoulli random variable with variance at most $P_{\max}$, results from random matrix theory (see e.g.~\cite{vers}) imply that whp for every $\ell \in [k]$
\begin{equation}
    \label{eq:rand}
    \sigma(C_\ell) \leq 2\sqrt{P_{\ell \ell}} \leq 4 \hsigma(F_\ell),  
\end{equation}
where the last inequality follows from Claim~\ref{cl:sbm1}. 

We now show that the sample means $\mu(C_\ell)$ and $\hmu(F_\ell)$ are close. 
\begin{claim}
\label{cl:meansbm}
The following event happens whp: for every $\ell \in [k]$, 
$|\mu(C_\ell) - \hmu(F_\ell)| \leq 5  \sqrt{P_{\max} \log n/w_0}.$
\end{claim}
\begin{proof}
Consider a coordinate $i$ where $v^i$ belongs to community $\ell'$. Then $(\hmu(F_\ell))_i = P_{\ell \ell'}$, and Bernstein's inequality implies that 
$$ \Pr \left[ \left| \frac{\sum_{x \in C_\ell} x_i}{|C_\ell|} - P_{\ell \ell'} \right| \geq 5 \sqrt{\frac{P_{\max} \log n}{w_0 n}} \right] $$ is at most $1/n^2$.
This shows that whp 
$|\mu(C_\ell) - \hmu(F_\ell)|$ is at most $10 \sqrt{\frac{P_{\max} \log n}{w_0}}$.
\qed
\end{proof}

Claim~\ref{cl:meansbm} and Claim~\ref{cl:sbm1} together imply that 
whp, for all distinct $\ell, \ell' \in [k]$
$$|\mu(C_\ell) - \mu(C_{\ell'})| \geq
\frac{1}{2} \cdot |\hmu(F_\ell) - \hmu(F_{\ell'})| \stackrel{Claim~\ref{cl:sbm1}}{\geq} 2\gamma \sqrt{P_{\max}} \stackrel{Claim~\ref{cl:sbm}}{\geq} \gamma \hsigma_0,$$
where $\hsigma_0$ denotes $\max_\ell \hsigma(F_\ell).$ This shows that clusters $C_\ell$ satisfy well-separatedness condition. It remains to show that \wntsc is satisfied. 

\begin{lemma}
\label{lem:wntsc}
With high probability, \wntsc is satisfied for all clusters $C_\ell, \ell \in [k]$. 
\end{lemma}

\begin{proof}
Fix an index $\ell$. Let $I_\ell$ be the coordinates $i$ corresponding to the cluster $C_\ell$. We know that $|I| = |C_\ell| \geq w_0 n$. For sake of brevity, let $n_\ell$ denote $|C_\ell|$. Let $v \in \Re^n$ be the unit vector with $v_i = \frac{1}{\sqrt{n_\ell}},$ if $i \in I_\ell$; 0 otherwise. We define a discrete probability distribution $\mu_\ell$ on the real line as follows: for a point $y \in \Re$, $$ \mu_\ell(y) := \Pr_{X \sim F_\ell} (X \cdot v = y). $$

The following property should be seen as anti-concentration property of $\mu_\ell$. 
\begin{claim}
\label{cl:anticsbm}
For every interval $I \subseteq \Re,$
$$ \mu_\ell(I) \leq \frac{|I|}{\sqrt{P_{\ell \ell}}} + \frac{1}{\sqrt{n_\ell P_{\ell \ell}}}.$$
\end{claim}
\begin{proof}
Observe that if $X \sim F_\ell$, then $X \cdot v = \frac{1}{\sqrt{n_\ell}} \sum_{i \in I_\ell} X_i. $
It follows that the maximum probability mass on any point is at most $q = \frac{1}{\sqrt{n_\ell P_{\ell \ell}}}.$ Since $X_i$ are either 0 or 1, $X \cdot v$ is integral multiple of $\frac{1}{\sqrt{n_\ell}}.$ Therefore, 
$$\mu_\ell(I) \leq  q \left( 1  + |I| \sqrt{n_\ell} \right)
\leq \frac{1}{\sqrt{n_\ell P_{\ell \ell}}} + \frac{|I|}{\sqrt{P_{\ell \ell}}}.$$
\qed
\end{proof}

Armed with the above anti-concentration result, \wntsc follows from similar arguments as in the proof of Theorem~\ref{stoc-appl}. Fix a subset $T \subseteq C_\ell,$
$|T| \geq \sqrt{n} \log n/100$ 
Let $I$ be the interval of length $\zeta$ (on both sides) around $\mu(T) \cdot v$, where $\zeta = \frac{|T| \sqrt{P_{\ell \ell}}}{12 |C_\ell|}.$ Claim~\ref{cl:anticsbm} implies that 
$$ \mu_\ell(I) \leq \frac{|T|}{6|C_\ell|} + \frac{1}{\sqrt{n_\ell P_{\ell \ell}}} \leq \frac{|T|}{5 |C_\ell|},$$
where the last inequality follows from the fact that 
$n_\ell P_{\ell \ell} \gg k^2/w_0^2$ (using~\eqref{eq:sbmc}).

Consider the set system consisting of intervals on the line, where the measure of each interval is given by $\mu_\ell$. Let $\varepsilon$ denote $\frac{|T|}{10 |C_\ell||}.$
Since $|T| \geq \sqrt{n} \ln n/100$, 
$$ |C_\ell| \geq \frac{c'}{\varepsilon^2} \ln \frac{k}{ \varepsilon}, $$
where $c'$ is a large enough constant (this follows from~\eqref{eq:sbmc}).
It follows that 
with probability at least $1 - \frac{o(1)}{k}$, $C_\ell$ is an $\varepsilon$-sample for this set system. 
 Therefore, 
$$\frac{|C_\ell \cap I|}{|C_\ell|} \leq \frac{|T|}{5 |C_\ell|}+ \frac{|T|}{10 |C_\ell|}. $$
Therefore, at  least $0.7|T|$ points of $T$ are at distance greater than $\zeta$ from $u \cdot \mu(T)$, showing that
$\sigma_L(T)\geq 0.49 \zeta$, where $L$ is the line along $v$. Using the definition of $\zeta$ and Claim~\ref{cl:sbm}, we see that $$ \sigma(T) \geq \sigma_L(T) \geq \frac{|T| \hsigma(F_\ell)}{25 |C_\ell|} \stackrel{\eqref{eq:rand}}{\geq} \frac{|T| \sigma(C_\ell)}{100 |C_\ell|} $$
This proves the \wntsc property for $C_\ell$. 
\qed
\end{proof}

\section{NP-hardness}
\label{sec:np}
The \iden problem is defined as follows: given a set of $n$ points $P$ in $\Re^d$, a target cardinality $h$, is there a subset $X$ of $P$, $|X|=h$, with $\sigma(X) \leq 1$ ? In this section, we prove the following:
\begin{theorem}
\label{thm:nphard}
Given a set of points $P$ and a parameter $h$, checking whether there is a subset $X$ of size $h$ with $\sigma(X) \leq 1$ is NP-complete. Further, the problem of finding the subset $X$ of size $h$ with the minimum $\sigma(X)$ value is APX-hard. 
\end{theorem}

The ideas in the reduction are similar to those in~\cite{CivrilM09}. 
We reduce from \ecov. An instance of \ecov is given by a set system $({\cal S}, U)$ consisting of a collection $\cal S$ of subsets of a ground set $U$. Let $m$ denote $|U|$. Each set in $\cal S$ has cardinality 3, and each element of $U$ appears in exactly 3 distinct sets in $\cal S$. The problem is to decide whether there is a sub-collection of $\cal S$ of size $m/3$ which covers all the elements in $U$. 

Given such an instance $\cI$ of \ecov, we reduce it to an instance $\cI'$ of \iden as follows: we define $m$ points in $\Re^n$, where we have a point $x(S)$ for each set $S \in {\cal S}$. We define the parameter $h$ to be $\frac{m}{3}$. 
If the set $S = \{i_1, i_2, i_3\}$, then 
$x(S)_i = \frac{\sqrt{h}}{\sqrt{3}}$, when $i=i_2, i_3, i_3$; 0 otherwise. 

\begin{lemma}
\label{lem:np}
$\inst'$ has a subset $X$ of cardinality $h$ satisfying $\sigma(X) \leq 1$ iff  $\inst$ has a set cover of size $h$. 
\end{lemma}
\begin{proof}
First suppose $\inst$ has an exact cover $\{S_1, \ldots, S_h\}$ of size $h$. 
Note that these sets must be mutually disjoint. Define $X$ to be the set of points $x(S_1), \ldots, x(S_h)$. 
Let $A$ be the $n \times h$ matrix whose columns are given by $x(S_1), \ldots, x(S_h)$. Since the columns of $A$ are orthogonal, and each of them has the same
length $\sqrt{h}$, it follows that $||A|| = \sqrt{h}$. Since $\sigma(X) \leq \frac{||A||}{\sqrt{h}} \leq 1$, one direction of the desired result follows. 

To show the converse, suppose $\inst$ has no set cover of size $h$. Let $X$ be a subset of $h$ points in $\inst'$. As above, let $A$ be the $n \times h$ matrix representing the coordinates of the points in $X$. 
The sets in $\cI$ corresponding to $X$  cannot be mutually disjoint, otherwise they will form a set cover in $\inst$. In other words, there must be two columns in $A$, say wlog column 1 and 2,  which have a non-zero value in the same row. Again, by renumbering, assume that $A_{11} = A_{12}  = \frac{\sqrt{h}}{\sqrt{3}}.$ Let $P$ denote the indices $(i,j)$ such that $A_{ij}$ is non-zero (i.e., equal to $\frac{\sqrt{h}}{\sqrt{3}}$. 

We are interested in the matrix $A'$ hose $i^{th}$ row is given by subtracting $\mu(X)$ from the $i^{th}$ row of $A.$ So we now the coordinates of $\mu(X)$. By definition of \ecov, each row in $A$ can have at most 3 non-zero coordinates. Therefore, $\mu(X)_i$ is at most $\frac{3}{h} \cdot \frac{\sqrt{h}}{\sqrt{3}} = \frac{\sqrt{3}}{\sqrt{h}}$ for $i=1, \ldots, n$. 
Therefore for every index $(i,j) \in P$, $A'_{ij} \geq \frac{\sqrt{h}}{\sqrt{3}} - \frac{\sqrt 3}{\sqrt h}. $ Further, for every pair $(i,j) \notin P, 1 \leq i \leq n, 1 \leq j \leq h,$ $A'_{ij} \geq -\sqrt{\frac{3}{h}}. $
Note that $\sigma(X) = \frac{||A'||}{\sqrt{h}}.$ We now argue that $||A'|| > \sqrt{h}$, which will then imply that $\sigma(X) > 1$.

Recall that $(1,1), (1,2) \in P$. Each of the first and the second columns of $A'$ has 3 non-zero entries. 
Two cases arise: 
\begin{itemize}
    \item There is no index $i  \in \{2, \ldots, n\}$ such that both $(i,1),(i,2)$ are in $P$. Since each column of $A$ has three non-zero entries, we can assume wlog that $(2,1),(3,1), (4,2), (5,2) \in P$.
     Consider the unit vector $v \in \Re^h$ with $v_1 = v_2 = \frac{1}{\sqrt{2}},$ and the other coordinates 0. A calculation shows that $||A'v|| \geq  \sqrt{h},$
     (assuming $h$ ius large enough). It follows that $\sigma(X) > 1$. 
    
    \item There is an index $i \in \{2, \ldots, n\}$ such that $(i,1), (i,2) \in P$: assume wlog that $i=2$. By renumbering rows of $A'$, we can also assume that $(3,1), (4,2) \in P$ (recall that each column of $A$ has exactly three non-zero entries).  Again, considering the unit vector $v$ as in the previous case, we see that 
    $||A'v|| \geq \sqrt h$, and so, $\sigma(X) > 1$ again. 
\end{itemize}
This proves the lemma. 
\qed
\end{proof}

It is easy to deduce from the reduction that the optimization version of \iden, which seeks to find a subset $X$ with minimum $\sigma(X)$,  is also APX-Hard. 

%
%
%
%
%
%
%
%
%
%
%
%
\section{Some Counter-Examples}
\label{sec:ex}

We first give an example showing that the elbow method can make a large error in estimate the value of $k$. For a set of points $X$, we $\Delta_k(X)$ to denote the optimal $k$-means cost, and let $k^\star$ denote (assume that the true number of clusters is at least 2)
$$ \argmax_{n-1 \geq k \geq 2} \frac{\Delta_{k-1}(X)}{\Delta_k(X)} $$
\begin{lemma}\label{elbow-method}
For any positive integers $r,k \geq 2,$ and large enough $d$, there is a mixture of $k$ standard Gaussians ${\cal N}({\bf \mu},I)$, each with weight $1/k$ with every pairwise mean separated by $k^r$, such that with high probability, for sufficiently large $n$ (polynomially bounded in $d,k$), $k^\star(X)  \neq k$, where $X$ is a sample of $n$ points from the mixture of Gaussians.
\end{lemma}

\begin{remark}\label{gap-statistic}
The Gap Statistic is akin to the elbow method, except it compares $\Delta_k$ on the data against the expected $\Delta_k$ on data generated from a single component null hypothesis mixture. In the simple example of Lemma (\ref{elbow-method}), with a suitable null hypothesis, one may get the correct $k$, but, no general results are known.
\end{remark}

\begin{proof}
For every integer $\ell \in [-k,k]$ define $\mu_k$ as the vector $\left( \frac{4 \sqrt{d}}{k} \ell, 0, 0, \ldots, 0 \right).$
We define a mixture with $2k+1$ components, where the component $F_\ell, \ell \in [-k,k]$ is ${\cal N}(\mu_k, I)$. Note that all component means lie on the first coordinate axis. Assuming the sample size $n$ is large enough, the sample means for each of the components also lie close this axis -- for sake of simplicity (though this assumption can be easily removed), assume that the sample means coincide with the corresponding component means. We also assume for simplicity (again, this can be easily removed) that there are exactly $\frac{n}{2k+1}$ samples from each component. 

Let $X$ be a sample of $n$ points, with $C_\ell$ being the points from $F_\ell$. Observe that for $\ell \in [-k,k]$, whp 
$$ \sum_{x \in C_\ell} |\mu(C_\ell) - x|^2 = |C_\ell| \cdot (d + O(\sqrt{d}). $$

 Consider a solution to the $k'$-means problem where we locate a set of $k'$ centers at $A = \{a_1, \ldots, a_{k'}\}$. For a point $x$, let $d(x,A)$ denote $\min_{a \in A} |x-a|$.
 Also suppose we assign all points in a cluster to a common center in $A$ (this may not be the best way of minimizing the $k'$-means objective function, but will give a tight enough upper bound). 
 Then, for a fixed cluster $C_\ell$, 

$$ \sum_{x \in C_\ell} d(x,A)^2  = \sum_{x \in C_\ell} |x - \mu(C_\ell)|^2 + |C_\ell| d(\mu(C_\ell), A)^2. $$
Now using the above inequality, it follows that the $k'$-means cost of the solution $A$ is at most: 
$$ n(d + O(\sqrt{d})) + \sum_{\ell \in [-k,k]}  |C_\ell| \cdot d(\mu(C_\ell), A)^2. $$

By symmetry, $\Delta_1(X)$ is achieved by placing a center at the origin, in which case (since all points of a cluster $C_\ell$ are assigned to the same center, we have equality here) 
$$\frac{\Delta_1(X)}{n} = d + O(\sqrt{d}) + \frac{16d \cdot \sum_{\ell \in [-1,1]} \ell^2 }{3} = 11.66 d + O(\sqrt{d}).$$

 Now we upper bound $\Delta_2(X).$ One solution is to locate two centers at on the first axis with coordinates $2 \sqrt{d}$ and $-2 \sqrt{d}$ respectively.
A routine calculation shows that 
$$ \frac{\Delta_2(X)}{n} \leq d + O(\sqrt{d}) + \frac{12d}{5} 
= 3.4d + O(\sqrt{d}). $$
 Therefore, $\frac{\Delta_1(X)}{\Delta_2(X)} \geq 
\frac{11.66}{3.4} - O \left( \frac{1}{\sqrt{d}} \right) \sim 3.42. $ As long as $k'$ stays at most $k$, $\frac{\Delta_{k'}(X)}{n} \geq d + \sqrt{d}. $
Therefore (using the upper bound on $\Delta_2(X)$) for all $k' \geq 3$, 
$$ \frac{\Delta_{k'-1}(X)}{\Delta_{k'}(X)} \leq 3.4.$$

It follows that if $k' \in [2,k]$, the highest ratio is achieved at $k'=2$. 
\qed
\end{proof}
We now show that for points drawn from GMM satsifying separation condition, the $k$-means and $2k$-means cost can be very close to each other, and so even a PTAS (with approximation ratio $(1+\varepsilon)$ where $\varepsilon > 0$ does not depend on the dimension $d$) may not be able to distinguish between these two costs. 

\begin{lemma}\label{k-means-example-2}
Suppose $F$ is a mixture of $2k$ standard Gaussians $F_1, \ldots, F_{2k}$ in $\Re^d$ with $d>100k^{25}$ and with uniform weights and for a large constant $\gamma$,
$$\mu(F_{2\ell})=\gamma (8 k^6,\ell k^8,0,0,\ldots ,0)\; ,\;
\mu(F_{2\ell-1})=
\gamma (-8k^6,
\ell k^8,0,0,\ldots ,0)\mbox{ for }\ell=1,2,\ldots ,k.$$
Then, for a sample $S$ of size $n$ from $F$, 
\begin{align}
&\E(\Delta_k(S))\leq d+O(\sqrt d)+100\gamma^2k^{12}\label{k-means-cost}\\
&\E(\Delta_{2k}(S))\geq d-O(\sqrt d)\label{2k-means-cost}
\end{align}
\end{lemma}

\begin{proof}
For (\ref{k-means-cost}), we note that for a standard Gaussian, with high probability,  the expected distance squared of a sample from the mean of the Gaussian is $d$ and with high probability it lies in $[d-c\sqrt d,c+\sqrt d]$. Further, if we choose $k$ centers as $\gamma (0,\ell k^8,0,0,\ldots ,0), \ell=1,2,\ldots ,k$, then, for samples from $F_{2\ell}$ and $F_{2\ell-1}$ the expected distance squared to 
$\gamma (0,\ell k^8,0,0,\ldots ,0)$ is $d+64\gamma^2k^{12},$ since we may choose the first coordinate of the sample independently of other coordinates. So, (\ref{k-means-cost}) follows.

For the $2k$ means cost, it is not difficult to see that choosing the $2k$ means of the $2k$ Gaussian densities is nearly optimal with high probability and this implies
(\ref{2k-means-cost}).
\qed
\end{proof}

We now give an example which shows that for any constants $c \geq 1, \varepsilon > 0$, an input set of points can be $(c, \varepsilon)$-stable with respect to two different values of the parameter $k$.

\begin{examp}
\label{stability}
The example is very simple, and also captures other deterministic conditions like proximity~\cite{KK10}. 
The dimension $d=2.$ There are $n$ points which are divided into 4 groups of size $n/4$ each -- call them $G_1, \ldots, G_4$. The points in these groups are co-located at $(D, 1), (D,-1), (-D,1), (-D,-1)$ respectively, where $D \gg n$. Assuming $n$ is large enough (compared to $(c,\varepsilon)$) the following  
clustering into 2 clusters is $(c, \varepsilon)$-stable: $G_1 \cup G_2, G_3 \cup G_4$. Indeed, the 2-means cost of this clustering is $n/4$. Now any clustering which differs from this on more than $\varepsilon n$ points has cost at least $\varepsilon n \cdot D^2 > cn/4$. Therefore, this clustering is $(c,\varepsilon)$-stable. But so is the clustering $(G_1, G_2, G_3, G_4)$, whose 4-means cost is 0. 

\end{examp}

This example shows that formulating the tightness definition in terms of 1-means cost does not suffice. 

\begin{examp}
\label{ex:2}
The data $X$ in $\Re^d$ is generated by a  GMM consisting of with two components (each being spherical Gaussian with unit variance), of weight 1/2 each.  The means of the two Gaussians is separated by a large constant $c$. Whp the average 1-means cost of the whole data is $d+O(\sqrt d)+(c^2/4)$.
For any $\varepsilon\in \Omega(1)$, any subset of $\varepsilon$ fraction of data can be seen to have average 1-means cost of at least $d-c\ln(1/\varepsilon)\sqrt d$, which is $1-o(1)$ of that of the whole data 
for $d\rightarrow\infty$. Thus 1-means cost is not a good measure of ``tightly packed''.

However, it is easy to see that $\sigma(X)$ $c/2$, whereas the $\sigma$ of the data generated by one component is at most $2$, thus, the \wntsc property is violated, indicating that $k>1$.
\end{examp}

{\small
\bibliographystyle{alpha}
\bibliography{paper}
}

\section*{Appendix}
\subsection{Missing proofs from Section~\ref{sec:prel}}

\simple*
\begin{proof}
There is a unit vector $v$ such that $|S|\sigma^(S)=\sum_{x\in S} (v\cdot (x-\mu(S)))^2$. Now, we have
$$\sum_{x\in S} (v\cdot (x-\mu(S)))^2\leq \sum_{x\in S} (v\cdot (x-\mu(X)))^2\leq \sum_{x\in X} (v\cdot (x-\mu(X)))^2
\leq\sigma^2(X)|X|,$$
proving the Claim.\qed

\end{proof}

\abc*
\begin{proof}
The 1-means cost of $X$ is at most $\opt(\cI)$ and at least $\opt(\cI)/4.$ Further, $\sigma(X)^2$ is at least the 1-means cost of $X$ and at most $d$ times this quantity (since $\sigma(X)^2$  is the maximum 1-means cost of $X$ along any direction). 
\qed
\end{proof}

\Tech*

\begin{proof}

Let $$v=\mu(S)-\mu(R).$$
\begin{align}
&\sum_{j\in R\cap S}(v\cdot (A_{\cdot,j}-\mu(R)))^2\leq 
\sum_{j\in R}(v\cdot (A_{\cdot,j}-\mu(R)))^2
\leq |v|^2|R|\sigma^2(R)   \label{197}
\end{align}
On the other hand, we have:
\begin{align}\label{198}
\sum_{j\in R\cap S}(v\cdot (A_{\cdot,j}-\mu(R)))^2\nonumber
&\geq\frac{1}{2}\sum_{j\in R\cap S}(v\cdot (\mu(S)-\mu(R)))^2\nonumber
 -\;  \sum_{j\in R\cap S}(v\cdot (A_{\cdot,j}-\mu(S)))^2\nonumber\\
&\geq \frac{1}{2}|R\cap S|\; |v|^4-\sum_{j\in S}(v\cdot (A_{\cdot,j}-\mu(S)))^2\nonumber\\
&\geq\frac{|R\cap S|}{2}|v|^4-|v|^2\sigma^2(S)|S|,
\end{align}
where,  first inequality uses the fact that $(a+b)^2\geq a^2/2-b^2$ for any reals $a,b$. 
The Lemma now follows from  
 (\ref{197}) and (\ref{198}). 
\hfill\qed
\end{proof}

\subsection{Missing proofs from Section~\ref{sec:id1}}

We prove the following corollary of Lemma~\ref{lem:clus1}.

\begin{restatable}{corollary}{kbound}
\label{cor:compare}
$\hatk \geq k$. 
\end{restatable}

\begin{proof}
Suppose not. 
Define a (partial) function $\tau : [\hatk] \rightarrow [k]$, where $\tau(\ell), \ell \in [\hatk]$ is the unique index $h$ (assuming it exists) with $|C_h \cap T_\ell| \geq \frac{w_0^2 n}{10}$ (Lemma~\ref{lem:clus1}). Since $k > \hatk$, there is an index $h$ such that $\tau^{-1}(h)$ is empty. In other words, 
%
$|C_h \cap T_\ell| \leq \frac{w_0^2 n}{10}$ for all $\ell \in \{1, \ldots, \hatk\}$. But then 
$$ |C_h| = \sum_{\ell=1}^\hatk |C_h \cap X_\ell| \leq \frac{w_0^2 \hatk n}{10} < \frac{w_0^2 k n}{10} < w_0 n, $$ which is a contradiction.
\qed
\end{proof}

\end{document}